\documentclass{layout_paper}

\graphicspath{{Fig/}}

\begin{document}

\title{Reducing Estimation Risk in Mean-Variance Portfolios with Machine Learning\footnote{I would like to thank Paul Ehling, Christian Brinch, Ragnar Juelsrud and Jo Saakvitne for valuable suggestions and comments.}}
\author{Daniel Kinn\footnote{Department of Economics, BI Norwegian Business School, Nydalsveien 37, N-0484 Oslo. E-mail: daniel.o.kinn@bi.no.}}
\date{July 2018}
\maketitle

\begin{abstract}
In portfolio analysis, the traditional approach of replacing population moments with sample counterparts may lead to suboptimal portfolio choices. I show that optimal portfolio weights can be estimated using a machine learning (ML) framework, where the outcome to be predicted is a constant and the vector of explanatory variables is the asset returns. It follows that ML specifically targets estimation risk when estimating portfolio weights, and that ``off-the-shelf'' ML algorithms can be used to estimate the optimal portfolio in the presence of parameter uncertainty. The framework nests the traditional approach and recently proposed shrinkage approaches as special cases. By relying on results from the ML literature, I derive new insights for existing approaches and propose new estimation methods. Based on simulation studies and several datasets, I find that ML significantly reduces estimation risk compared to both the traditional approach and the equal weight strategy.\\

\end{abstract}

\section{Introduction}\label{sec_pf_intro}
In the modern portfolio theory framework developed by \cite{markowitz1952portfolio} the optimal portfolio is a function of the population mean and covariance matrix of asset returns. Given data on returns, the traditional approach is to estimate optimal portfolio weights by treating the sample mean and sample covariance matrix as if they were the true population moments. Figure \ref{fig_pf_intro} illustrates the use of this strategy on a random sample of 20 assets from Standard \& Poor's 500 Index (S\&P500). Clearly, the out of sample return is highly volatile at the end of the sample period, which can be traced back to large asset positions. Extreme asset weights and poor out of sample performance are well documented shortcomings of the traditional approach, see e.g. \cite{black1992global}, \cite{best1991sensitivity} and \cite{jorion1985international}. A plausible explanation is estimation risk; the fact that sample moments may be imprecise estimates of the population moments.\footnote{Contrary, \cite{green1992will} argue that extreme portfolio positions may exist in the population due to a dominant factor in equity returns. In that case the instability in Figure \ref{fig_pf_intro} could be caused by a dominant factor, not estimation risk. The instability could also be due to a misspecified portfolio.}

\par The effect of estimation risk on portfolio choice has been recognised since \cite{klein1976effect}, who showed that the optimal portfolio choice differs from the traditional choice in the presence of uncertain parameters. It is important to note that the estimation risk problem is not only a feature of small samples. For a given number of observations, estimation risk is increasing in the number of assets of the portfolio. The empirical study by \cite{demiguel2007optimal} suggests that estimation risk is large even when the portfolio size is modest and estimated based on five years of monthly observations. 

\begin{figure}
\centering
\includegraphics[scale=0.6]{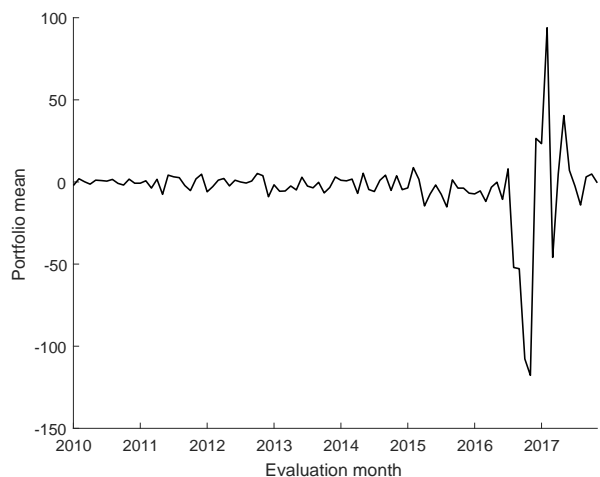}
\caption{\footnotesize \textbf{Illustration of estimation risk.} Monthly out of sample portfolio return based on the traditional approach. Each portfolio return is computed based on portfolio weights estimated using the previous 120 months. The data is a random sample of 20 stocks from the S\&P500 from 2000 to 2017.}
\label{fig_pf_intro}
\end{figure}

\par In this paper I show that optimal portfolio weights can be estimated in a machine learning (ML) framework.\footnote{See e.g. \cite{hastie2011elements} and \cite{murphy2012machine} for excellent discussions of machine learning. I will use machine learning as a general term, but will specifically be referring to supervised machine learning where the function to be estimated is linear.} Broadly speaking, the ML framework may be thought of as a penalized regression problem where the regressors are the asset returns, the coefficients are the portfolio weights, the outcome to be predicted is constant and a penalty is imposed to avoid large weights.\footnote{This formulation implies a mean-variance perspective in order to derive the relationship between ML methods and portfolio theory. Much of the existing literature focus on the minimum variance portfolio. This focus is in large justified by the difficulty in estimating means, see e.g. \cite{jorion1985international}. Both \cite{jagannathan2003risk} and \cite{demiguel2009generalized} argue that nothing much is lost by ignoring the mean altogether. On the other hand, \cite{jorion1986bayes} use simulation to show that for modest sample sizes, mean-variance approaches outperform the minimum variance portfolio. Furthermore, \cite{demiguel2007optimal} show that equal weighting, which incorporates both moments, seldom is outperformed by strategies ignoring the mean.} Estimating optimal portfolio weights with the ML framework has three important implications.

\par First, estimating portfolio weights with the ML framework is equivalent to choosing a portfolio in order to minimize total risk, which is the sum of the risk inherent in the optimal (population) portfolio and the estimation risk. This result follows because under quadratic utility, total risk is equivalent to an expected out of sample mean squared error, which is the minimization objective of ML algorithms. This mean squared error is simply the expected squared difference between a constant and the portfolio return. By using sample splitting, portfolio weights are estimated on one sample and tested on another. For assets showing unstable sample moments across subsamples (thus being subject to large estimation risk), an estimated weight based on one subsample might generalise poorly to other subsamples in terms of the mean squared error. By imposing a bound (penalty) on the weight, the out of sample mean squared error can be improved, reducing total risk and consequently the estimation risk.

\par Second, estimation risk may be decomposed into a bias-variance tradeoff. This decomposition is common in ML, and I discuss why the tradeoff is important in the portfolio context. For instance, the portfolio weights derived from the traditional approach are unbiased, but are likely to exhibit large variance in repeated samples of returns. In contrast, a passive strategy where asset weights are fixed to be equal may lead to large bias, but such weights do not vary in repeated samples of returns. In between, ML seeks to choose portfolio weights to balance bias with variance in order to minimize estimation risk. \cite{jagannathan2003risk} recognizes this tradeoff, but refers to it as a tradeoff between specification error and sampling error. 

\par Third, ``off-the-shelf'' ML methods can be used to estimate the optimal portfolio weights. These methods offer standardised ways of doing cross-validation, estimation and assessment of performance.\footnote{\cite{ban2016machine} also study ML for portfolio optimization, but focus on a particular method they label performance-based regularization, and not ``off-the-shelf'' ML methods. They use cross-validation for determining the weight penalty, and customize the procedure so that it targets the Sharpe ratio and bounds the range of possible penalty values. They document that the approach works well for risk reduction in small portfolios. In contrast, my focus is larger portfolios and standard ML algorithms where cross-validation and penalty bounds are implemented using standard software.} I use ``off-the-shelf'' ML to shed light on existing approaches and to introduce new methods for portfolio estimation. The main findings are discussed below.

\par The traditional approach is a special case of the ML framework, equivalent to a regression problem without penalty (OLS). I argue that the OLS formulation provides an alternative explanation for why the traditional approach is associated with large estimation risk. Since OLS is the best linear \emph{unbiased} estimator, the traditional strategy does not allow for a tradeoff between bias and variance. Thus, for large portfolios, the traditional approach may show large estimation risk due to overfitting in a regression sense.

\par Imposing constraints may reduce the overfitting problem. One approach is L1 regularization, where the sum of the absolute value of the portfolio weights is required to be smaller than some threshold. \cite{demiguel2009generalized} proposed to add this constraint to the weights of the minimum variance portfolio. They showed that a special case recovers the no short selling portfolio analysed by \cite{jagannathan2003risk}. In a mean-variance setting, \cite{demiguel2007optimal} pointed out that L1 regularization is equivalent to shrinking the expected returns towards the average return. My results elaborate on the mean-variance setting. Since L1 is equivalent to an ``off-the-shelf'' ML method known as Lasso, a well known relationship between Lasso and OLS can be used to derive new insights. In essence, L1 regularization implies that the weight estimates from the traditional approach (OLS) are shrunk by the same amount. For assets where this amount is larger than the traditional estimate, the Lasso weight is set to zero.\footnote{Both \cite{fan2012vast} and \cite{brodie2009sparse} study Lasso for portfolio estimation. However, they do not make this connection to the traditional approach. Furthermore, they do not use cross-validation to determine the penalty level. \cite{demiguel2009generalized} use cross-validation, but limits attention to the minimum variance portfolio.}

\par Another approach is L2 regularization, implying that the sum of squared portfolio weights must be smaller than a threshold. For the minimum variance portfolio, \cite{demiguel2009generalized} show that L2 regularization is equivalent to shrinking the covariance matrix towards the identity matrix, similar to the approach by \cite{ledoit2004honey,ledoit2004well}. In the mean-variance setting, L2 regularization is equivalent to Ridge regression by \cite{hoerl1970ridge}. Using standard ML results I show that Ridge regression shrinks the traditional weights by the same factor and that for penalty values in a specific range, Ridge regression outperforms the traditional approach in terms of estimation risk. It follows that Ridge regression outperforms the traditional approach if the optimal portfolio is well diversified. 

\par I introduce two other ``off-the-shelf'' ML methods for portfolio estimation; Principal Component regression and Spike and Slab regression. The former assumes that the asset returns are generated from some low dimensional model, such as e.g. a factor model. The idea is to estimate the portfolio weights using all assets, but only use the variation in returns that is attributable to the lower dimensional model. The size of the low dimensional model is determined by cross-validation. 

\par Spike and Slab is a Bayesian variable selection technique for linear regression. By assuming a Bernoulli prior (the ``Spike''), the approach uses a binary rule for including or excluding assets from the regression, i.e. the portfolio. Conditional on an asset being included, a Gaussian prior (the ``Slab'') is assumed for the portfolio weight. The Spike and Slab formulation results in a posterior for the included assets and a posterior for the portfolio weights, conditional on included assets only. By using Gibbs sampling, I draw from these posteriors several thousand times, resulting in an inclusion probability of each asset and posterior portfolio weight distributions. The Spike and Slab approach to portfolio selection bears some resemblance to the Empirical Bayes approach.\footnote{\cite{jorion1985international,jorion1986bayes} shrinks the sample mean of each asset toward the portfolio mean of the global minimum variance portfolio.} I show that the posterior portfolio weights are a combination of the portfolio weights from the traditional approach and the Gaussian prior means, conditional on included assets only. Thus, like Empirical Bayes, weights can be expressed as a combination of the sample mean and a prior, but unlike Empirical Bayes, the attention is limited to a subset of assets.

\par Based on simulation, I find that ML algorithms significantly improve on the traditional approach and several benchmark strategies, including the mean-variance portfolio without short selling, the minimum variance portfolio and the equally weighted portfolio. Consistent with the existing literature, I find that the constraints imposed by these benchmark strategies may work for small sample sizes. However, the strict nature of the constraints (disallowing negative weights, ignoring means or equal weighting) may be harmful in larger samples, where the sample moments are likely to be more precisely estimated. In contrast, ML algorithms impose ``softer'' constraints, in the sense that they will cause the estimated portfolio weights to converge to the traditional weights as the sample size grows. This is beneficial because as the sample size approaches infinity, the traditional approach is clearly the best choice. 

\par Finally, I apply ML to several different real world datasets. Using assets from the S\&P500, I find that the discussed ML algorithms yield similar out of sample Sharpe ratios, significantly outperforming the traditional approach, the minimum variance portfolio and the equal weight strategy. Using industry portfolios where each asset is a combination of stocks, I document similar results for the ML methods, the minimum variance portfolio and the equal weight strategy, suggesting that the estimation risk problem is not as pronounced for this type of portfolios. I also consider a dataset covering 200 cryptocurrencies. The estimation risk is expected to be large due to the short lifetime, but large number of such currencies. When the number of parameters exceeds the number of observations, both the traditional approach and the minimum variance portfolio are infeasible due to a degenerate covariance matrix. I find that the ML algorithms yield similar Sharpe ratios to the equal weighted strategy in this case.

\par The paper is organised as follows. Section \ref{sec_pf_framework} briefly introduces ML to the unfamiliar reader and lays out the framework connecting ML to portfolio theory. In Section \ref{sec_pf_methods}, I discuss existing approaches from a ML perspective and introduce new methods. I assess the performance of ML for reducing estimation risk based on artificial data calibrated to the U.S. stock market in Section \ref{sec_pf_sim}. Finally, I apply ML to several different datasets, including the S\&P500, industry portfolios and a cryptocurrency portfolio in Section \ref{sec_pf_empirical}. Detailed derivations of the propositions and equations in this paper can be found in Appendix \ref{sec_pf_appendix_derivations}.

\section{Machine Learning in a Portfolio Context}\label{sec_pf_framework}
\subsection{A Brief Introduction to Machine Learning}\label{sec_pf_framework_ml}
 Let $y$ be some outcome variable drawn from the model 
\begin{equation}\label{eq_pf_framework_model}
y = f(\x) + \varepsilon
\end{equation}
where $\x$ is a $m$-dimensional vector of covariates and $\varepsilon$ is normally distributed with mean zero and variance $\phi^2$. The objective of ML is to learn the function $f$ in order to predict future values of $y$. A ML algorithm $q$ outputs an estimate $\hat{f}_q$ of $f$ based on a training set of data, $\mathcal{T}=\{y_i,\x_i\}_{i=1}^n$. How well $\hat{f}_q$ predicts new values of $y$ can be evaluated using a mean squared error, which I will refer to as the \emph{generalisation error}
\begin{equation}\label{eq_pf_framework_mse0}
F(\hat{f}_q) = \E_{(y_0,\x_0)}[(y_0-\hat{f}_q(\x_0))^2]
\end{equation}
where $(y_0,\x_0)$ is a new observation not used for training, and the expectation is with respect to the distribution of this new observation. Taking the expectation across training sets gives the \emph{expected generalisation error} of algorithm $q$
\begin{equation}\label{eq_pf_framework_mse}
F_q = \E_{\mathcal{T}}[F(\hat{f}_q)]=\E_{\mathcal{T}}\left\{\E_{(y_0,\x_0)}[(y_0-\hat{f}_q(\x_0))^2]\right\}
\end{equation}
where the expectation $\E_\mathcal{T}$ is with respect to the training set that produced the function $\hat{f}_q$. In other words, the expected generalisation error of algorithm $q$ is the squared loss from drawing infinitely many training sets, estimating $f$ using algorithm $q$ on each set, and evaluating on an infinitely large test set. To gain further insight on the performance of the algorithm, it is common to decompose (\ref{eq_pf_framework_mse}) into a tradeoff between bias and variance
\begin{equation}\label{eq_pf_framework_bv}
F_q = \E_\mathcal{T}\left\{\E_{y_0}[(y_0-\hat{f}_q(\x_0))^2]\right\} =\underbrace{(f(\x_0)-\E_\mathcal{T}[\hat{f}_q(\x_0)])^2}_{\text{squared bias}} + \underbrace{\V_\mathcal{T}[\hat{f}_q(\x_0)]}_{\text{variance}} + \underbrace{\phi^2}_{\text{noise}}
\end{equation}
where it is assumed that $\x_0$ is non-random for simplicity. If $\hat{f}_q$ is estimated based on a simple, underfitted algorithm, predictions will be biased because the true function value $f(\x_0)$ could be far from the expected prediction in repeated training sets, $\E_\mathcal{T}[\hat{f}_q(\x_0)]$. However, an overly simplistic model will give similar predictions in repeated draws of the data, so the variance is likely to be low. In contrast, a flexible, overfitted algorithm may lead to low bias, but the variance is likely to be large because we can expect predictions to vary substantially in repeated data draws. Having both low bias and low variance (avoiding both over- and underfitting) translates into a low expected generalisation error. In the following I show that this general framework and its intuition is transferable to the problem of estimation risk in mean-variance portfolio analysis.

\subsection{Generalisation Error and Estimation Risk}\label{sec_pf_framework_estrisk}
I consider a framework where the investor has preferences described by the utility function $u(r)$ over some random portfolio return $r$. As is common in portfolio theory, I limit the attention to a case where the utility function is completely characterised by the first two moments of the return.\footnote{There is a large literature showing that mean-variance preferences not necessarily coincide with the expected utility framework, see e.g. \cite{levy1979approximating}.} I will assume that the agent has mean-variance preferences described by the quadratic utility function
\begin{equation}\label{eq_pf_framework_u}
u(r) = r - \frac{1}{2}\alpha r^2
\end{equation}
where $\alpha>0$ is the risk aversion coefficient and $r$ is the portfolio return. The use of quadratic utility is necessary to show the equivalence with ML, but as is shown below, the derived relative portfolio weights in the risky assets correspond exactly to the weights obtained from the more common exponential utility function.

\par The agent can invest $\theta_f\in\R$ in a risk free asset with a given return $r_f$ and invest $\boldsymbol{\theta}\in\R^m$ in $m$ risky assets with \emph{excess} return over the risk free asset given by the vector $\x$. I assume that $\x$ is multivariate normal with expected excess return $\boldsymbol{\mu}$ and covariance matrix $\boldsymbol{\Sigma}$. Using the constraint that the asset positions in the risk free and the risky assets must sum to one, $\theta_f+\mathbf{1}'\boldsymbol{\theta}=1$ where $\mathbf{1}$ is a vector of ones, the expected quadratic utility of the agent is given by
\begin{equation}\label{eq_pf_framework_eu}
\E_{\x}[u(r_f + \x'\boldsymbol{\theta})]
\end{equation}
The following proposition states that expected quadratic utility (\ref{eq_pf_framework_eu}) may be written as a generalisation error of the form (\ref{eq_pf_framework_mse0}), implying that  maximizing expected quadratic utility is equivalent to minimizing the generalisation error.
\begin{prop}\label{prop_pf_framework_generror} Generalisation error. Maximizing expected utility with a quadratic utility function is equivalent to minimizing the generalisation error
\begin{equation}\label{eq_pf_framework_generror}
F(\boldsymbol{\theta}) = \E_\x[(\bar{r}-\x'\boldsymbol{\theta})^2]
\end{equation}
where $\bar{r} = (1-\alpha r_f)/\alpha$.
\end{prop}
Proposition \ref{prop_pf_framework_generror} has two implications. First, it can be used to derive the optimal portfolio weights. If the distribution of excess returns is known, $\x\sim\N(\boldsymbol{\mu},\boldsymbol{\Sigma})$, then minimizing (\ref{eq_pf_framework_generror}) with respect to $\boldsymbol{\theta}$ gives the optimal (population) portfolio weights
\begin{equation}\label{eq_pf_framework_thetaPop}
\boldsymbol{\theta}^* = (\boldsymbol{\Sigma} + \boldsymbol{\mu}\boldsymbol{\mu}')^{-1} \boldsymbol{\mu} \bar{r}
\end{equation}
and the corresponding minimum generalisation error $F_*=F(\boldsymbol{\theta}^*)$. Second, the proposition provides a link between the objective of ML and the portfolio problem. Suppose that the return distribution is unknown to the agent, but that he uses an ML algorithm $q$ to estimate the optimal portfolio weights based on empirical data. Let the outcome variable be constant $y=\bar{r}$ and let $f$ be the linear function $\x'\boldsymbol{\theta}$. It follows that the portfolio problem of choosing $\boldsymbol{\theta}$ to maximize expected quadratic utility may be viewed as the ML problem of estimating $\boldsymbol{\theta}$ to minimize an \emph{expected generalisation error} similar to (\ref{eq_pf_framework_mse})
\begin{equation}\label{eq_pf_framework_mlgenerror}
F_q= \E_\mathcal{T}\{\E_{\x_0}[(\bar{r}-\x_0'\hat{\boldsymbol{\theta}}_q)^2]\}
\end{equation}
where the expectation $\E_\mathcal{T}$ is with respect to the training data that was used to obtain the estimate $\hat{\boldsymbol{\theta}}_q$ and $\E_{\x_0}$ is with respect to the out of sample returns $\x_0\sim\N(\boldsymbol{\mu},\boldsymbol{\Sigma})$. It may be hard to see how it in practice is possible to minimize $F_q$ as it includes the distribution of the out of sample returns in addition to infinitely many draws of training data. However, $F_q$ can be approximated by cross-validation, see e.g. \cite{hastie2011elements}. Furthermore, $F_q$ may be written as the sum of the minimum generalisation error $F_*$ and an estimation risk component $R_q$
\begin{equation}
F_q = F_* + R_q
\end{equation} 
From this definition it follows that any ML algorithm minimizing the expected generalisation error $F_q$ is equivalently minimizing estimation risk $R_q$, since $F_*$ is an irreducible population value. In theory, if some algorithm $q$ truly minimizes the expected generalisation error such that $F_q=F_*$, it also maximizes expected quadratic utility. Proposition \ref{prop_pf_framework_estrisk} examines $R_q$ in closer detail.
\begin{prop}\label{prop_pf_framework_estrisk}Estimation risk. The estimation risk $R_q$ of ML algorithm $q$ is the difference between the expected generalisation error, $F_q$, and the minimum generalisation error, $F_*$, giving
\begin{equation}\label{eq_pf_framework_estrisk}
R_q = \underbrace{(\boldsymbol{\theta}^*-\E_\mathcal{T}[\hat{\boldsymbol{\theta}}_q])'\mathbf{A}(\boldsymbol{\theta}^*-\E_\mathcal{T}[\hat{\boldsymbol{\theta}}_q])}_\text{squared bias} + \underbrace{\text{tr}\left(\mathbf{A}\V_\mathcal{T}[\hat{\boldsymbol{\theta}}_q]\right)}_\text{variance}
\end{equation}
where $\mathbf{A}=\boldsymbol{\Sigma}+\boldsymbol{\mu}\boldsymbol{\mu}'$ and $\text{tr}(.)$ denotes the trace operator.
\end{prop}
\par Proposition \ref{prop_pf_framework_estrisk} highlights that the estimation risk of any ML strategy $q$ can be decomposed into a bias-variance tradeoff. As such, the intuition from ML carry over to the portfolio problem. Compared to the optimal weights, an ``underfitted portfolio'' where only a few assets receive non-zero weights will show relatively low variance out of sample due to relatively low exposure. However, bias can be substantial, as letting several assets have zero weights may forego investment opportunities that are present in the optimal portfolio. Contrary, an ``overfitted portfolio'' consisting of a large set of assets could give lower bias, but the variance is likely to be high in repeated samples of returns, due to a large set of parameters that needs to be estimated from the data. Thus, minimizing the sum of squared bias and variance is instrumental for obtaining low levels of estimation risk. 

\par Estimation risk is non-negative, $R_q\geq 0$ as long as $\mathbf{A}$ is positive semidefinite, with $R_q=0$ for $\hat{\boldsymbol{\theta}}_q=\boldsymbol{\theta}^*$. Intuitively, since the optimal portfolio weights $\boldsymbol{\theta}^*$ obtain the minimum generalisation error $F_*$, any estimator $\hat{\boldsymbol{\theta}}_q$ different from the optimal portfolio weights will provide a higher expected generalisation error $F_q$ and thus positive estimation risk. 

\begin{figure}
\begin{minipage}{1\textwidth}
\begin{center}
\begin{tikzpicture}[thick,scale=0.7, every node/.style={scale=0.7}]
\draw[thick,->] (0,0) -- (0,7) node[above]{$\mu$};
\draw[thick,->] (0,0) -- (7,0) node[below]{$\sigma$};
\draw[thick] (2,4) to [bend right=45](4,7);
\draw[thick] (1,1) to [bend left=35](5,4);
\draw[thin] (0,0) -- (5,7);
\draw[thin,dotted] (0,5) -- (3.5,5);
\draw[thin,dotted] (3.5,0) -- (3.5,5);
\draw[thin,dotted] (0,2.28) -- (1.6,2.28);
\draw[thin,dotted] (1.6,0) -- (1.6,2.28);
\draw[thin,dotted] (0,6.5) -- (3.5,4.95);
\draw (4,7) node[left]{$F_*$};
\draw (0,5) node[left]{$\mu^*$};
\draw (3.5,0) node[below]{$\sigma^*$};
\draw (0,2.28) node[left]{$\mu_p^*$};
\draw (1.6,0) node[below]{$\sigma_p^*$};
\draw (0,6.5) node[left]{$\bar{r}$}; 
\draw (2.8,6.0) node[left]{$\sqrt{F_*}$}; 
\draw (3.5,4.95) node[circle,fill,scale=0.2,label=above:$A$]{A};
\draw (1.6,2.28) node[circle,fill,scale=0.2,label=above:$A'$]{A'};
\end{tikzpicture}
\end{center}
\end{minipage}
\begin{minipage}{1\textwidth}
\caption{\footnotesize{\textbf{The optimal portfolio}. The objective of the agent is to minimize generalisation error in order to get a portfolio return as close as possible to the ideal return $\bar{r}$, which can be interpreted as some certain, maximum level of return. Minimizing generalisation error (maximizing expected quadratic utility) gives the optimal portfolio weights $\boldsymbol{\theta}^*$ located at $A$, with corresponding mean $\mu^*=\boldsymbol{\mu}'\boldsymbol{\theta}^*$, standard deviation $\sigma^*=(\boldsymbol{\theta}^{*'}\boldsymbol{\Sigma}\boldsymbol{\theta}^*)^{1/2}$ and minimum generalisation error $F_*$. This solution gives the minimum distance from $\bar{r}$ to $A$, which has length given by $\sqrt{F_*}$ where $F_*=(\bar{r}-\mu^*)^2 + (\sigma^*)^2$. The portfolio of optimal relative weights (the tangency portfolio) is located at $A'$, with portfolio mean $\mu_p^*=\boldsymbol{\mu}'\boldsymbol{\omega}^*$ and standard deviation $\sigma_p^*=(\boldsymbol{\omega}^{*'}\boldsymbol{\Sigma}\boldsymbol{\omega}^*)^{1/2}$, tangent to the efficient portfolio frontier. The optimal weights at $A$ has the same Sharpe ratio as the tangency portfolio at $A'$, indicated by the slope of the line through the origin.}}
\label{fig_pf_framework_sketch_pop}
\end{minipage}
\end{figure}
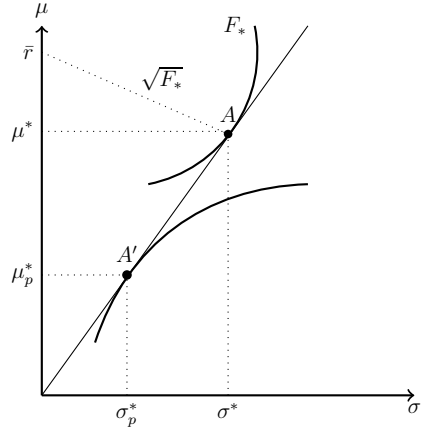

\par The results in this section are illustrated in Figure \ref{fig_pf_framework_sketch_pop}, where the optimal portfolio weights $\boldsymbol{\theta}^*$ are located at point $A$, with mean $\mu^*=\boldsymbol{\mu}'\boldsymbol{\theta}^*$ and standard deviation $\sigma^*=\sqrt{\boldsymbol{\theta}^{*'}\boldsymbol{\Sigma}\boldsymbol{\theta}^*}$. The curved line through point $A$ shows combinations of mean and standard deviation that achieves the minimum generalisation error, $F_*$, and thus maximum expected utility. At the optimal portfolio weights at point $A$, the minimum generalisation error is given by $F_*=\E_\x[(\bar{r}-\x'\boldsymbol{\theta}^*)^2]=(\bar{r}-\mu^*)^2+(\sigma^*)^2$. Thus, by the Pythagorean theorem, the line segment connecting $\bar{r}$ on the vertical axis with point $A$ has length given by the square root of the minimum generalisation error. In other words, the optimal portfolio weights minimize the distance from the certain, maximum obtainable return $\bar{r}$ to point $A$.

\par The portfolio at $A'$ is computed using the optimal \emph{relative} weights in the $m$ risky assets\footnote{Equation (\ref{eq_pf_framework_tangencyPop}) is derived based on quadratic utility, but the exact same result holds for the exponential utility case, see e.g. \cite{demiguel2007optimal}. However, the use of quadratic utility is necessary to establish the link between utility and generalisation error in Proposition \ref{prop_pf_framework_generror}.}
\begin{equation}\label{eq_pf_framework_tangencyPop}
\boldsymbol{\omega}^* = \frac{\boldsymbol{\theta}^*}{\mathbf{1}'\boldsymbol{\theta}^*} = \frac{\boldsymbol{\Sigma}^{-1}\boldsymbol{\mu}}{\mathbf{1}'\boldsymbol{\Sigma}^{-1}\boldsymbol{\mu}}
\end{equation}
which is the well known expression for the tangency portfolio, with corresponding portfolio mean $\mu_p^*=\boldsymbol{\mu}'\boldsymbol{\omega}^*$ and standard deviation $\sigma_p^*=\sqrt{\boldsymbol{\omega}^{*'}\boldsymbol{\Sigma}\boldsymbol{\omega}}$. According to portfolio theory, the tangency portfolio $\boldsymbol{\omega}^*$ maximizes the Sharpe ratio $\mu_p^*/\sigma_p^*$, which is the slope of the line from the origin through $A'$, tangent to the efficient portfolio frontier. In fact, this observation also tells us that $\boldsymbol{\theta}^*$ has the same optimal Sharpe ratio as the tangency portfolio $\boldsymbol{\omega}^*$. Indeed,
\begin{equation}
\frac{\mu_p^*}{\sigma_p^*}=\frac{\boldsymbol{\mu}'\boldsymbol{\omega}^*}{(\boldsymbol{\omega}^{*'}\boldsymbol{\Sigma}\boldsymbol{\omega}^*)^{1/2}} = \frac{\boldsymbol{\mu}'\boldsymbol{\theta}^*c}{(c^2\boldsymbol{\theta}^{*'}\boldsymbol{\Sigma}\boldsymbol{\theta}^*)^{1/2}}=\frac{\mu^*}{\sigma^*}
\end{equation}
where $c=1/\mathbf{1}'\boldsymbol{\theta}^*$. Hence, both $A$ and $A'$ lie on the same line through the origin. In summary, the optimal portfolio weights that minimize the generalisation error has the same Sharpe ratio as the tangency portfolio. This has important practical implications, because it is straightforward to compare ML portfolios that estimate $\boldsymbol{\theta}^*$ to methods that estimate $\boldsymbol{\omega}^*$ by comparing their Sharpe ratios.

\subsection{The Machine Learning Portfolio}\label{sec_pf_framework_sample}
The results in propositions \ref{prop_pf_framework_generror} and \ref{prop_pf_framework_estrisk} suggest that any supervised ML algorithm can be used for portfolio selection, as the target of such algorithms is the expected generalisation error. I mainly restrict the discussion in this paper to linear ML models of the form
\begin{equation}\label{eq_pf_framework_ml}
\argmin_{\boldsymbol{\theta}} \left\{\frac{1}{n}\sum_{i=1}^n (\bar{r}-\x_i'\boldsymbol{\theta})^2\right\} \text{ subject to } P(\boldsymbol{\theta}) \leq s
\end{equation}
where $P(\boldsymbol{\theta})$ is some penalty function and $s$ is some threshold estimated from the data. Consider first the penalty $P(\boldsymbol{\theta})=\sum_{j=1}^m I(\theta_j\neq 0)$. Here $I(\theta_j\neq 0)$ is an indicator function taking the value one if asset $j$ receives a non-zero weight and zero otherwise. In this case $s$ may be interpreted as the maximum number of assets allowed in the risky portfolio. This formulation addresses estimation risk as follows. Suppose we set $s=2$ and solve (\ref{eq_pf_framework_ml}). Ignoring several assets that possibly have non-zero optimal weights would lead to a relatively low variance component, but possibly high bias. Contrary, crowding the portfolio with a large set of assets, e.g. $s=100$, could ensure that no important assets are left out, but possibly lead to high variance. The objective is thus to choose $s$ in order to minimize expected generalisation error, and thus also estimation risk.

\par The specific penalty above is known as best subset selection, see e.g. \cite{james2013introduction}. With 20 assets there are over one million different portfolios to choose from, and best subset is therefore computationally too expensive for large portfolios. Broadly speaking, many ML algorithms provide approximations to the best subset problem through different specifications of the penalty function. I will elaborate on different choices of $P(\boldsymbol{\theta})$ in Section \ref{sec_pf_methods}, but restrict the discussion below to choice of the tuning parameter $s$. For that purpose, it is common to rewrite (\ref{eq_pf_framework_ml}) as
\begin{equation}\label{eq_pf_framework_ml_lambda}
\hat{\boldsymbol{\theta}}_q = \argmin_{\boldsymbol{\theta}} \left\{\frac{1}{n}\sum_{i=1}^n (\bar{r}-\x_i'\boldsymbol{\theta})^2 + \lambda P(\boldsymbol{\theta})\right\}
\end{equation}
where $\hat{\boldsymbol{\theta}}_q$ are the estimated optimal portfolio weights from using algorithm $q$, and the algorithm is determined by what type of penalty $P(\boldsymbol{\theta})$ that is used. Furthermore, $\lambda$ is the Lagrange multiplier associated with the penalty and I discuss how to choose $\lambda$ in order to minimize expected generalisation error below.

\begin{figure}
\begin{minipage}{1\textwidth}
\begin{subfigure}{0.50\textwidth}
\begin{center}
\begin{tikzpicture}[thick,scale=0.7, every node/.style={scale=0.7}]
\draw[thick,->] (0,0) -- (0,7) node[above]{$\mu$};
\draw[thick,->] (0,0) -- (7,0) node[below]{$\sigma$};
\draw[thin] (0,0) -- (2.5,7);
\draw[thin] (0,0) -- (5,7);
\draw[thin,dotted] (0,6.5) -- (2.15,6);
\draw[thin,dotted] (0,3.3) -- (1.2,3.3);
\draw[thin,dotted] (1.2,0) -- (1.2,3.3);
\draw[thin,dotted] (0,2.28) -- (1.6,2.28);
\draw[thin,dotted] (1.6,0) -- (1.6,2.28);
\draw (0,6.5) node[left]{$\bar{r}$};
\draw (0,3.3) node[left]{$\hat{\mu}_p$};
\draw (1.2,0) node[below]{$\hat{\sigma}_p$};
\draw (0,2.28) node[left]{$\mu_p$};
\draw (1.6,0) node[below]{$\sigma_p$};
\draw[thick] (0.5,4.4) to [bend right=40](2.2,7);
\draw[thick] (1,2.5) to [bend left=23](2.5,5);
\draw (3.5,4.95) node[circle,fill,scale=0.2,label=above:$A$]{A};
\draw (2.15,6) node[circle,fill,scale=0.2,label=above right:$B$]{B};
\draw (1.2,3.3) node[circle,fill,scale=0.2,label=above left:$B'$]{B'};
\draw (1.6,2.28) node[circle,fill,scale=0.2,label=above:$A'$]{A'};
\end{tikzpicture}
\caption{In sample}
\label{fig_pf_framework_sketch_sample_train}
\end{center}
\end{subfigure}\hspace*{\fill}
\begin{subfigure}{0.50\textwidth}
\begin{center}
\begin{tikzpicture}[thick,scale=0.7, every node/.style={scale=0.7}]
\draw[thick,->] (0,0) -- (0,7) node[above]{$\mu$};
\draw[thick,->] (0,0) -- (7,0) node[below]{$\sigma$};
\draw[thin] (0,0) -- (7,4.5);
\draw[thin] (0,0) -- (5,7);
\draw (3.2,6.2) node[left]{$\sqrt{F_\text{OLS}}$};
\draw (2.4,5.2) node[left]{$\sqrt{F_*}$}; 
\draw (5.4,4.2) node[left]{$\sqrt{R_\text{OLS}}$}; 
\draw[thin,dotted] (0,6.5) -- (6.8,4.4);
\draw[thin,dotted] (0,6.5) -- (3.5,4.95);
\draw[<->,thin] (3.5,4.95) -- (6.8,4.4);
\draw[thin,dotted] (0,2.28) -- (1.6,2.28);
\draw[thin,dotted] (1.6,0) -- (1.6,2.28);
\draw[thin,dotted] (0,3) -- (4.65,3);
\draw[thin,dotted] (4.65,0) -- (4.65,3);
\draw (0,6.5) node[left]{$\bar{r}$};
\draw (0,3) node[left]{$\tilde{\mu}_p$};
\draw (4.65,0) node[below]{$\tilde{\sigma}_p$};
\draw (0,2.28) node[left]{$\mu_p$};
\draw (1.6,0) node[below]{$\sigma_p$};
\draw (3.5,4.95) node[circle,fill,scale=0.2,label=below:$A$]{A};
\draw (6.8,4.4) node[circle,fill,scale=0.2,label=above:$B$]{B};
\draw (4.65,3) node[circle,fill,scale=0.2,label=above:$B'$]{B'};
\draw (1.6,2.28) node[circle,fill,scale=0.2,label=above:$A'$]{A'};
\end{tikzpicture}
\caption{Population}
\label{fig_pf_framework_sketch_sample_test}
\end{center}
\end{subfigure}
\caption{\footnotesize{\textbf{The traditional approach and estimation risk.} The figures illustrate the problem of estimation risk for the traditional approach applied to a portfolio with $m$ assets, assuming only one training set. The population solution from Figure \ref{fig_pf_framework_sketch_pop} is located at $A$ and $A'$ in both figures above. Figure \ref{fig_pf_framework_sketch_sample_train}: In sample, the OLS estimated weights $\hat{\boldsymbol{\theta}}$ yield the solution at $B$. This solution is closer to $\bar{r}$ than the population solution at $A$ due to overfitting and thus a low in-sample error. The in-sample tangency portfolio $\hat{\boldsymbol{\omega}}$ is located at $B'$, with corresponding in-sample moments $\hat{\mu}_p = \hat{\boldsymbol{\mu}}'\hat{\boldsymbol{\omega}}$ and $\hat{\sigma}_p = (\hat{\boldsymbol{\omega}}'\hat{\boldsymbol{\Sigma}}\hat{\boldsymbol{\omega}})^{1/2}$, respectively. Figure \ref{fig_pf_framework_sketch_sample_test}: Point $B$ shows the OLS weights evaluated at the population values. Clearly, $B$ is further from the ideal return $\bar{r}$ than the population solution at $A$, which is due to the fact that the overfitted OLS weights $\hat{\boldsymbol{\theta}}$ generalise poorly. The distance from $\bar{r}$ to $A$ and $B$ is given by $\sqrt{F_*}$ and $\sqrt{F_\text{OLS}}$, respectively, where $F_\text{OLS} = (\bar{r}-\boldsymbol{\mu}'\hat{\boldsymbol{\theta}})^2 + \hat{\boldsymbol{\theta}}'\boldsymbol{\Sigma}\hat{\boldsymbol{\theta}}$. Thus, the length from the optimal solution at $A$ to the OLS solution at $B$ is given by the square root of the estimation risk, $\sqrt{R_\text{OLS}}$. The out of sample tangency portfolio is located at $B'$ with portfolio mean $\tilde{\mu}_p = \boldsymbol{\mu}'\hat{\boldsymbol{\omega}}$ and portfolio standard deviation $\tilde{\sigma}_p = (\hat{\boldsymbol{\omega}}'\boldsymbol{\Sigma}\hat{\boldsymbol{\omega}})^{1/2}$. In terms of Sharpe ratio, indicated by the lines through the origin, the out of sample tangency portfolio at $B'$, makes the agent worse off than the population tangency portfolio $A'$, due to estimation risk.}}
\label{fig_pf_framework_sketch_sample}
\end{minipage}
\end{figure}

\subsubsection{When $\lambda=0$: The Traditional Approach}
Ordinary least squares (OLS) is an important special case of (\ref{eq_pf_framework_ml_lambda}) where the penalty parameter is $\lambda=0$. Solving the problem in this case yields
\begin{equation}\label{eq_pf_framework_thetaSample}
\hat{\boldsymbol{\theta}} =(\X'\X)^{-1}\X'\y
\end{equation}
where $\X$ is the $n\times m$ data matrix of returns and $\y=\mathbf{1}\bar{r}$ is a constant $n\times 1$ vector. Note that (\ref{eq_pf_framework_thetaSample}) is the sample counterpart to the optimal weights in (\ref{eq_pf_framework_thetaPop}), i.e. $\hat{\boldsymbol{\theta}} = (\hat{\boldsymbol{\Sigma}}+\hat{\boldsymbol{\mu}}\hat{\boldsymbol{\mu}}')^{-1}\hat{\boldsymbol{\mu}}\bar{r}$, where the sample mean is $\hat{\boldsymbol{\mu}} = \frac{1}{n} \X'\mathbf{1}$ and the maximum likelihood sample covariance is $\hat{\boldsymbol{\Sigma}} = \frac{1}{n}(\X-\mathbf{1}\hat{\boldsymbol{\mu}}')'(\X-\mathbf{1}\hat{\boldsymbol{\mu}}')=\frac{1}{n}(\X'\X)-\hat{\boldsymbol{\mu}}\hat{\boldsymbol{\mu}}'$. The estimated optimal relative weights $\hat{\boldsymbol{\omega}} = \hat{\boldsymbol{\theta}}/\mathbf{1}'\hat{\boldsymbol{\theta}}$ then correspond to the sample version of the tangency portfolio. Hence, the OLS solution is equivalent with the traditional approach to portfolio estimation, where sample moments are used directly in the Markowitz formulae. The connection between OLS and the sample counterpart to the theory provides a way of discussing the shortcomings of the traditional approach with regards to estimation risk.

\par It is well known that OLS provides the best linear unbiased estimator, see e.g. \cite{hayashi2000econometrics}, but the OLS solution can in many cases be severely overfitted, so that predictions generalise poorly to new data. The traditional approach may be thought of as an OLS problem, leading to low bias but possibly overfitting, thus providing a poor generalisation error. In other words, the traditional approach only offers minimum variance conditional on bias being zero, which is not the same as minimizing the sum of squared bias and variance. 

\par An illustration of the estimation risk problem for the traditional approach is provided in Figure \ref{fig_pf_framework_sketch_sample}. Consider Figure \ref{fig_pf_framework_sketch_sample_train}, where point B represents the in-sample solution from the traditional approach based on a large portfolio of $m$ assets. Using the OLS analogy, estimating a regression with many assets is likely leading to overfitting unless the training data is very large. The implication is a low in-sample root mean squared error, measured as the distance from $\bar{r}$ to $B$, lower than the population error measured from $\bar{r}$ to $A$. The intuition is that the sample solution is highly flexible, using the $m$ parameters to fit spurious patterns in the estimated sample moments, leading to a too low in-sample error. Figure \ref{fig_pf_framework_sketch_sample_test} shows the estimation risk of the traditional approach when the OLS solution is applied to the population moments (i.e. out of sample). The spurious patterns picked up in sample does not generalise to the population, leading to a high generalisation error measured from $\bar{r}$ to $B$ as $\sqrt{F_\text{OLS}}$. The distance from $B$ to $A$ is therefore given by the square root of the estimation risk, $\sqrt{R_\text{OLS}}$. Clearly, the traditional approach yields positive estimation risk in this case.

\subsubsection{When $\lambda>0$: Approximating the Generalisation Error}
Compared to OLS, any ML algorithm with a positive penalty $\lambda>0$ introduces bias in the portfolio weight estimates, but as long as the decrease in variance is larger than the increase in squared bias, the expected generalisation error will decrease, thereby also leading to lower estimation risk. 

\par The problem however is how to choose $\lambda$ in order to minimize expected generalisation error $F_q$. This error is unobserved, but ML algorithms approximate it by heuristic sample splitting techniques such as \emph{$K$-fold cross-validation}. The training set of returns is randomly assigned to $K$ subsamples or ``folds'' without replacement. Let $\mathcal{I}_k$ denote the index set of returns assigned to fold $k$ and $\mathcal{I}_{-k}$ the index set of the remaining returns not assigned to fold $k$. For a given algorithm $q$ and a given value of $\lambda$, estimate (\ref{eq_pf_framework_ml_lambda}) using the return data in $\mathcal{I}_{-k}$ and denote the estimated portfolio vector by $\hat{\boldsymbol{\theta}}_{q,\mathcal{I}_{-k}}$. Test the estimated portfolio on the hold-out fold using the mean squared error 
\begin{equation}
\hat{F}_q^k(\lambda) = \frac{1}{|\mathcal{I}_k|}\sum_{i\in\mathcal{I}_k} (\bar{r}-\x_i'\hat{\boldsymbol{\theta}}_{q,\mathcal{I}_{-k}})^2
\end{equation}
where $|\mathcal{I}_k|$ denotes the number of observations in fold $k$. Repeat this process for each fold $k=1,\hdots,K$ and compute the average error $\hat{F}_q(\lambda) = \frac{1}{K}\sum_{k=1}^K \hat{F}_q^k(\lambda)$. Then $\lambda$ can be chosen to minimize this error, $\lambda^*=\argmin_\lambda \hat{F}_q(\lambda)$, and the optimal ML portfolio $\hat{\boldsymbol{\theta}}_q$ is obtained by estimating (\ref{eq_pf_framework_ml_lambda}) on all of the training data with $\lambda=\lambda^*$.

\par If cross-validation works well, the estimated weights $\hat{\boldsymbol{\theta}}_q$ will, when applied to out of sample returns, provide an expected generalisation error that is close to the true value $F_q$ given in (\ref{eq_pf_framework_mlgenerror}). For some algorithm $q$ estimated with a positive penalty level, we should expect that $F_q<F_\text{OLS}$. By the estimation risk formula (\ref{eq_pf_framework_estrisk}), it thus follows that $R_\text{OLS}>R_q>R_*=0$. In terms of Figure \ref{fig_pf_framework_sketch_sample_test}, this would result in a ML portfolio on a line from the origin lying between A and B, providing lower estimation risk than the traditional approach. In terms of the in-sample situation in Figure \ref{fig_pf_framework_sketch_sample_train}, we would expect a larger training error compared to the traditional solution at B.

\par Dividing the data into folds has an intuitive explanation. If the first and second moment of the returns are unstable across folds, it is taken as evidence that the sample moments are imprecise estimates of the population moments, and hence that estimation risk is large. To illustrate how cross-validation mitigate the problem, suppose that the number of folds is $K=2$, and that the mean of a particular asset $j$ is relatively high with a relatively low variance in the first fold. Based on this fold it would be optimal to invest a positive amount in this asset. However, if the second fold shows the opposite case; a low mean and a high variance for the same asset, the positive weight will generalise poorly to the second fold. In general, exposure to assets that show unstable moments across folds will result in a high out of sample mean squared error on the left out fold. Shrinking the exposure to such assets by increasing $\lambda$ will reduce estimation risk. Obviously, asset moments will stabilize across folds as the number of observations increases, so that as $n\rightarrow \infty$, the optimal penalty level will approach zero. As such, any ML algorithm of the form (\ref{eq_pf_framework_ml_lambda}) will approach the traditional approach, which again approaches the optimal portfolio weights.

\begin{figure}
\centering
\begin{minipage}{1\textwidth}
\begin{subfigure}{0.50\textwidth}
\includegraphics[width=\linewidth]{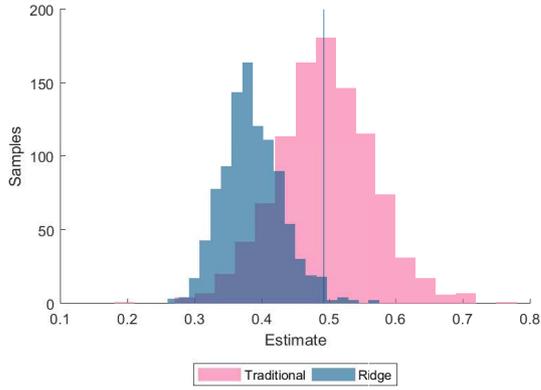}
\caption{Estimates of $\theta_1^*$}
\label{fig_pf_ml_hist}
\end{subfigure}\hspace*{\fill}
\begin{subfigure}{0.50\textwidth}
\includegraphics[width=\linewidth]{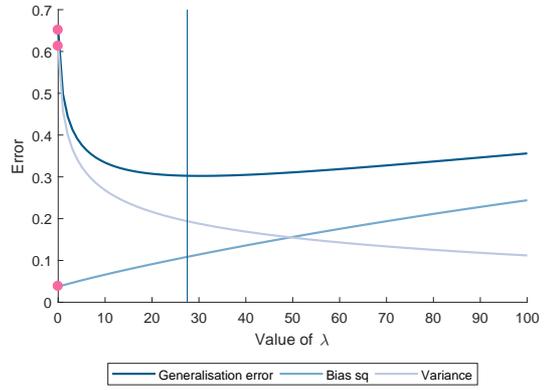}
\caption{Bias-Variance}
\label{fig_pf_ml_biasvariance}
\end{subfigure}
\medskip
\begin{subfigure}{0.50\textwidth}
\includegraphics[width=\linewidth]{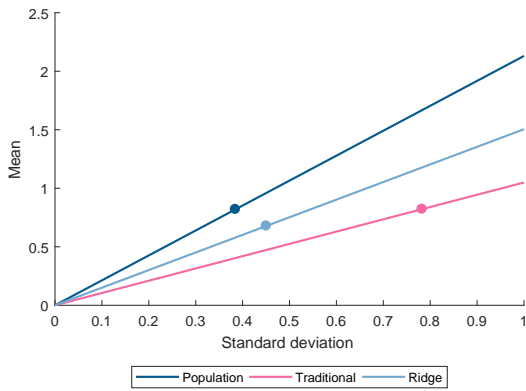}
\caption{Estimation Risk}
\label{fig_pf_ml_risk}
\end{subfigure}\hspace*{\fill}
\begin{subfigure}{0.50\textwidth}
\includegraphics[width=\linewidth]{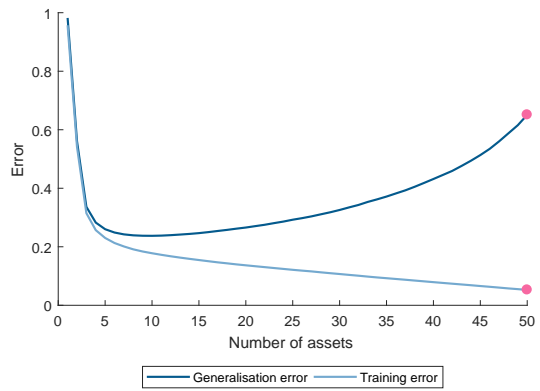}
\caption{Training vs Generalisation Error}
\label{fig_pf_ml_traintest}
\end{subfigure}\\
\caption{{\footnotesize \textbf{The traditional approach vs machine learning}. Figure \ref{fig_pf_ml_hist}: estimates of $\theta_1^*$ from the traditional approach and Ridge regression with 5-fold cross validation based on 1000 repeated draws of the training data. Figure \ref{fig_pf_ml_biasvariance}: decomposition of the expected generalisation error into squared bias and variance for varying values of $\lambda$. Figure \ref{fig_pf_ml_risk}: average out of sample Sharpe ratios evaluated at the population moments, similar to Figure \ref{fig_pf_framework_sketch_sample_test}. Figure \ref{fig_pf_ml_traintest}: training and expected generalisation error as a function of the number of assets in the portfolio, ranging from 1-50.}}
\label{fig_pf_ml}
\end{minipage}
\end{figure}

\subsubsection{Simulation Study: Machine Learning vs Traditional Approach}
As an illustration of the difference between the traditional approach and ML, I use simulated returns of $m=50$ assets from a multivariate normal distribution. The expected excess return vector $\boldsymbol{\mu}$ is assumed to be equal to zero for all but the three first assets. The population covariance $\boldsymbol{\Sigma}$ is constructed so that all assets are highly positively correlated, with a correlation coefficient of $0.95$. Based on this specification, the optimal portfolio $\boldsymbol{\theta}^*$ shows large positive investments in the two first assets and relatively small positions in the remaining assets. In general it is difficult to estimate the small positions in any finite sample, and ignoring or shrinking them may improve out of sample performance.

\par I use $n=70$ periods and estimate the optimal portfolio using the traditional approach (\ref{eq_pf_framework_thetaSample}) and one particular ML approach, Ridge regression, see Section \ref{sec_pf_methods_ridgelasso} for details. Based on repeated draws of the training data I obtain 1000 estimates of the portfolio weights from both approaches.

\par Figure \ref{fig_pf_ml_hist} shows histograms of all 1000 weight estimates of the first asset for both approaches. For Ridge regression, each estimate was obtained by using 5-fold cross validation to obtain the optimal penalty level. Clearly, the ML approach introduces bias in the estimate of $\theta_1^*$, but reduces variance. Figure \ref{fig_pf_ml_biasvariance} plots the expected generalisation error for Ridge regression together with its bias-variance decomposition for different values of the tuning parameter $\lambda$. The traditional approach is shown as red dots along the vertical axis, corresponding to the situation where $\lambda=0$. Increasing the tuning parameter leads to a reduction in variance at the cost of increasing bias (lowering return), but the overall error is reduced for penalty values below some threshold indicated by the vertical line. The threshold is reported as the average optimal $\lambda$-value across all training sets. Figure \ref{fig_pf_ml_risk} illustrates the estimation risk of each approach, similar to the sketch in Figure \ref{fig_pf_framework_sketch_sample_test}. In this study the ML approach reduces estimation risk by 70\% compared to the traditional approach.

\par The best subset idea is illustrated in Figure \ref{fig_pf_ml_traintest}, where the expected generalisation error of the traditional approach is plotted together with the mean squared error in-sample (training error) for a varying number of included assets. Underfitted portfolios with too few assets will generalise poorly. Increasing the number of assets reduces the generalisation error, but including too many assets results in overfitting and poor generalisation. The traditional approach with 50 assets show the lowest training error and thus the highest in-sample Sharpe ratio, but this ratio is not sustainable out of sample.

\section{Machine Learning Algorithms}\label{sec_pf_methods}
In this section I discuss ``off-the-shelf'' ML algorithms for portfolio estimation. In Section \ref{sec_pf_methods_ridgelasso}, I shed new light on existing shrinkage approaches by relating them to the ML literature. I introduce new tools for portfolio estimation in Section \ref{sec_pf_methods_pcr} and Section \ref{sec_pf_methods_ss}.

\subsection{Ridge and Lasso Regression}\label{sec_pf_methods_ridgelasso}
Ridge regression is a special case of (\ref{eq_pf_framework_ml_lambda}) obtained by specifying the penalty as $P(\boldsymbol{\theta})=\sum_{j=1}^m \theta_j^2$. The penalty is continuous and differentiable, leading to a closed form solution for the portfolio weights. Solving (\ref{eq_pf_framework_ml_lambda}) with the quadratic penalty gives the Ridge regression estimator
\begin{equation}\label{eq_pf_framework_thetaRidge}
\hat{\boldsymbol{\theta}}_{\text{R}}=(\X'\X+\lambda\I)^{-1}\X'\y
\end{equation}
and the relative weights $\hat{\boldsymbol{\omega}}_\text{R} = \hat{\boldsymbol{\theta}}_\text{R}/\mathbf{1}'\hat{\boldsymbol{\theta}}_\text{R}$. It follows from standard results, see e.g. \cite{efron2016computer}, that the Ridge regression solution (\ref{eq_pf_framework_thetaRidge}) is equivalent to the traditional approach (\ref{eq_pf_framework_thetaSample}), with an adjusted covariance matrix $\hat{\boldsymbol{\Sigma}}_\text{R}=\hat{\boldsymbol{\Sigma}} + \frac{\lambda}{n} \I$.\footnote{Plugging $\hat{\boldsymbol{\mu}}$ and $\hat{\boldsymbol{\Sigma}}_\text{R}$ into (\ref{eq_pf_framework_thetaPop}) gives $(\hat{\boldsymbol{\Sigma}}_\text{R} +\hat{\boldsymbol{\mu}}\hat{\boldsymbol{\mu}}')^{-1}\hat{\boldsymbol{\mu}}\bar{r}=(\hat{\boldsymbol{\Sigma}}+\frac{\lambda}{n}\mathbf{I} +\hat{\boldsymbol{\mu}}\hat{\boldsymbol{\mu}}')^{-1}\hat{\boldsymbol{\mu}}\bar{r}=(\frac{1}{n}\X'\X + \frac{\lambda}{n}\mathbf{I})^{-1}\frac{1}{n}\X'\mathbf{1}\bar{r}=(\X'\X + \lambda\mathbf{I})^{-1}\X'\y=\hat{\boldsymbol{\theta}}_\text{R}$.}  Thus, Ridge regression leaves the asset sample means unchanged, but increases the variance of each asset by a constant amount $\lambda/n$. It immediately follows that also the tangency portfolio $\hat{\boldsymbol{\omega}}_\text{R}$ will be based on the sample mean and the adjusted covariance matrix, $\hat{\boldsymbol{\Sigma}}_\text{R}$. Using a well known relationship between Ridge regression and OLS, I establish the following result.
\begin{prop}\label{prop_pf_framework_Ridge}
Ridge regression yields lower estimation risk than the traditional approach for penalty values $\lambda\in(0,\bar{\lambda})$, where $\bar{\lambda} = 2F_*/{\sum_{j=1}^m (\theta_j^*)^2}$.
\end{prop}
\par Proposition \ref{prop_pf_framework_Ridge} provides an intuitive theoretical result that has practical relevance. If the optimal portfolio has low return and a high standard deviation ($F_*$ is large) and/or if the optimal portfolio is diversified (low $\sum_{j=1}^m (\theta_j^*)^2$), it is more likely that Ridge regression may obtain a penalty value that will outperform the traditional approach in terms of estimation risk. 

\par Lasso regression offers a different way of shrinking the asset positions by specifying the penalty as $P(\boldsymbol{\theta})=\sum_{j=1}^m |\theta_j|$. The nature of the shrinkage is not as straightforward as for Ridge regression, since the penalty is non-smooth with no closed form solution for the weights. However, the non-smoothness gives Lasso an asset selection property. In practice, Lasso may set several of the $\theta_j$'s equal to zero so that the estimated portfolio will be based on a small subset of the assets. Thus, Lasso is a computationally cheap method close in spirit to best subset selection. In the case of orthogonal returns ($\X'\X=\mathbf{I}$) there exists a simple relationship between the traditional approach and Ridge and Lasso
\begin{equation}
\hat{\theta}_{q,j} = \begin{cases}
\hat{\theta}_j/(1+\lambda) &\text{ if Ridge}\\
\text{sign}(\hat{\theta}_j)(|\hat{\theta}_j|-\lambda)_+ &\text{ if Lasso}
\end{cases}
\end{equation}
where $\text{sign}(\hat{\theta}_j)$ is used to denote the sign of the traditional estimate and $(|\hat{\theta}_j|-\lambda)_+$ outputs the difference $|\hat{\theta}_j|-\lambda$ if it is positive and zero otherwise. In both cases the traditional weight estimates are shrunk towards zero, thereby giving a more conservative exposure to the different assets. For Ridge regression, weights are shrunk by the same factor, while for Lasso regression, each asset position is reduced by a constant amount, truncating at zero for assets where $\lambda$ is large enough compared to the traditional estimate. In other words, Ridge regression shrinks asset positions towards zero, but Lasso will in many cases shrink asset positions all the way to zero. As a result, Lasso will tend to outperform Ridge if several assets truly have optimal weights of zero, while Ridge may be the best choice if the optimal portfolio is highly diversified.

\begin{figure}
\centering
\begin{minipage}{1\textwidth}
\begin{subfigure}{0.32\textwidth}
\includegraphics[width=\linewidth]{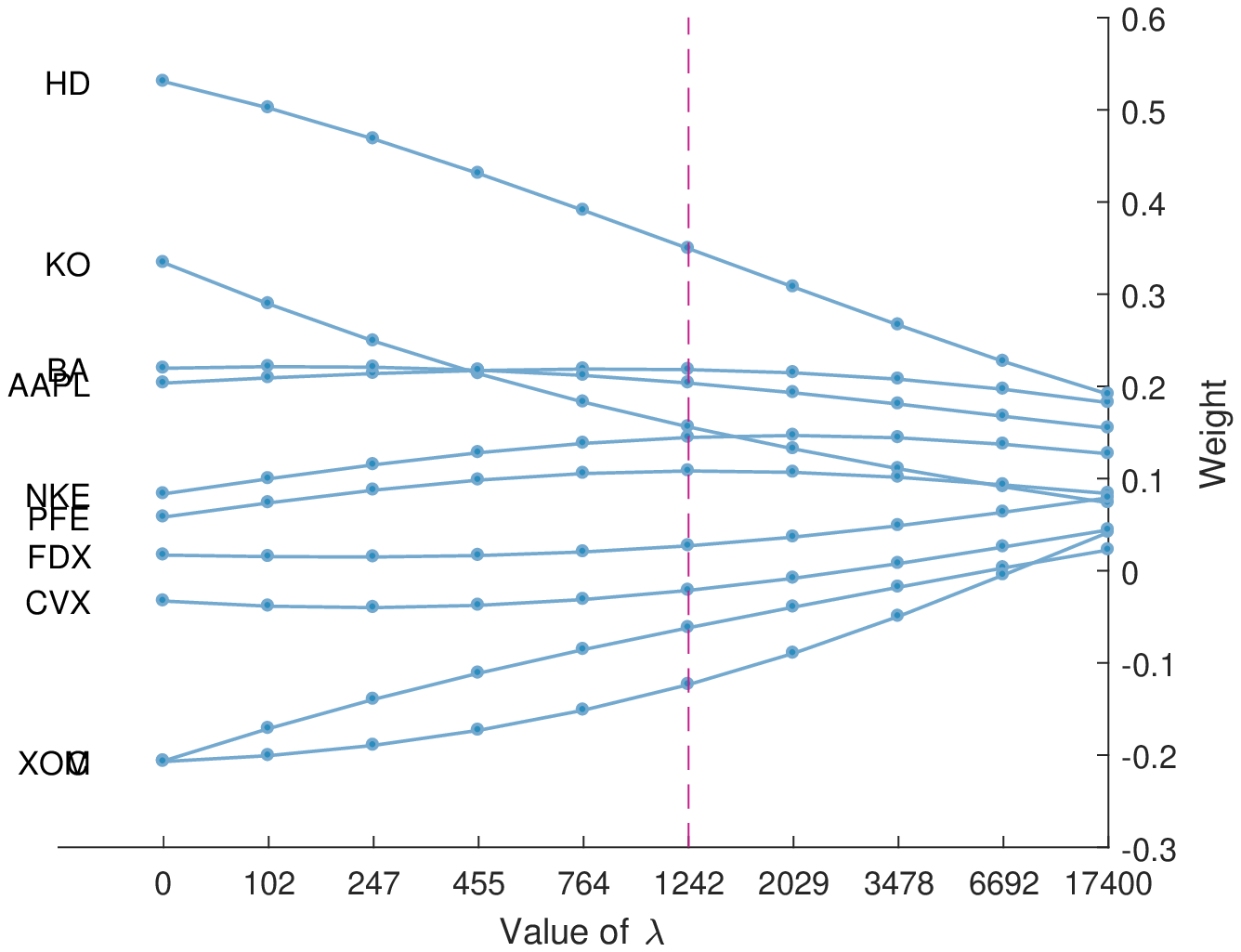}
\caption{Ridge}
\label{fig_pf_regpath_ridge}
\end{subfigure}\hspace*{\fill}
\begin{subfigure}{0.32\textwidth}
\includegraphics[width=\linewidth]{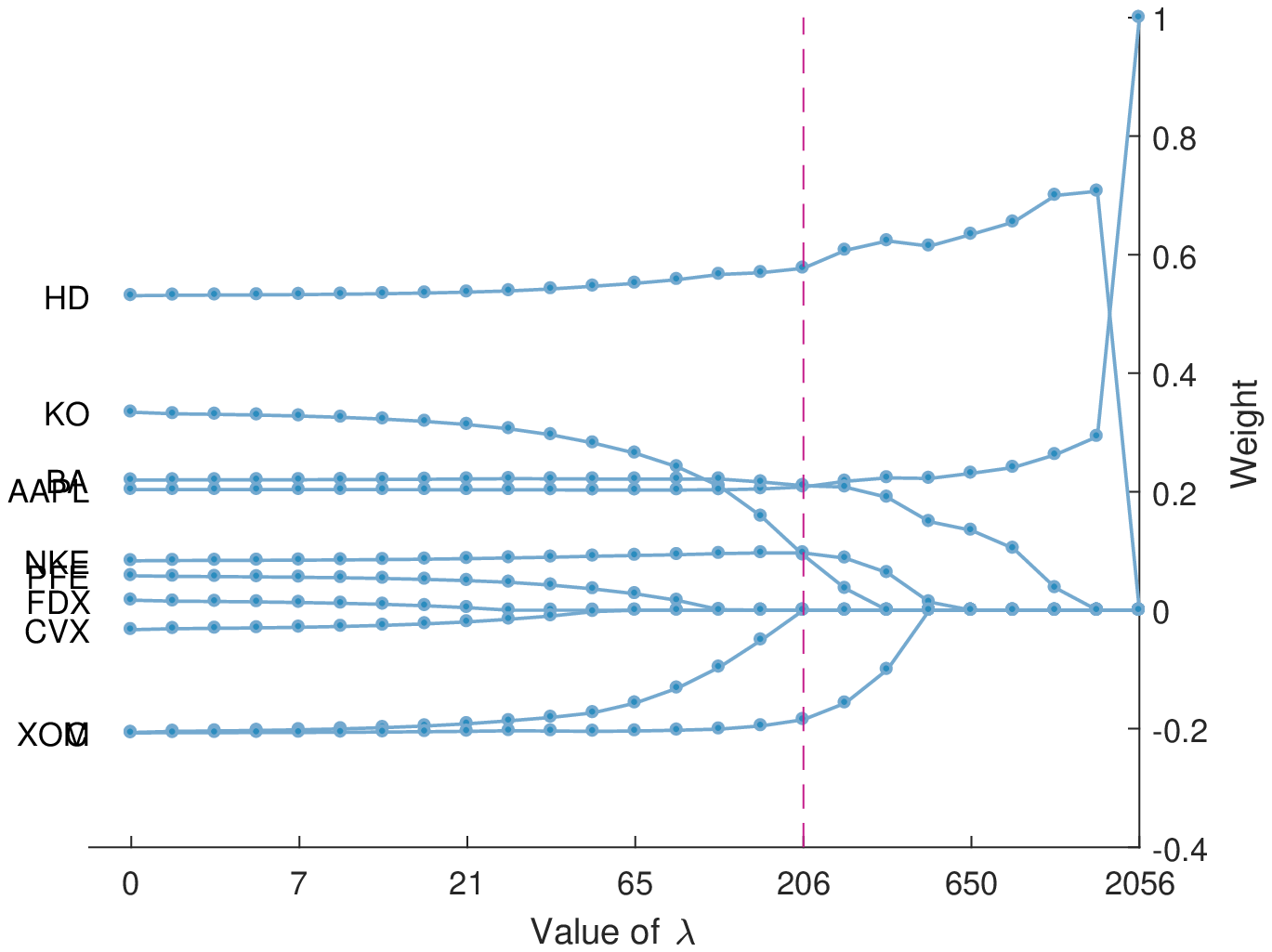}
\caption{Lasso}
\label{fig_pf_regpath_lasso}
\end{subfigure}\hspace*{\fill}
\medskip
\begin{subfigure}{0.32\textwidth}
\includegraphics[width=\linewidth]{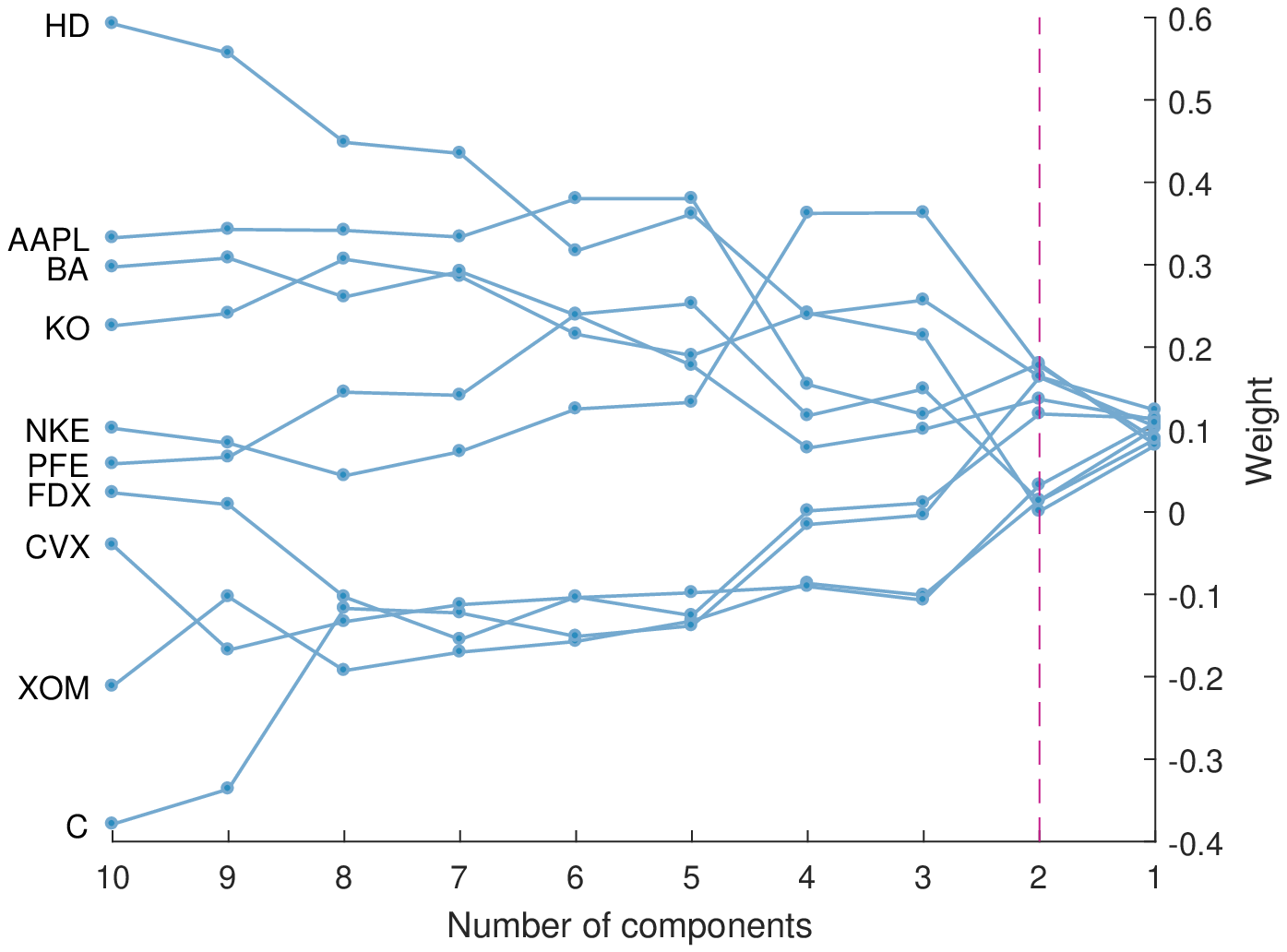}
\caption{PCR}
\label{fig_pf_regpath_pcr}
\end{subfigure}\\
\caption{\footnotesize \textbf{Portfolio regularization paths.} Regularization paths are computed based on monthly returns for 10 stocks from the S\&P500 from August 2012 to December 2017. Figure \ref{fig_pf_regpath_ridge} plots the continuous regularization paths of Ridge regression. Figure \ref{fig_pf_regpath_lasso} shows the case of Lasso regression where some assets are truncated at zero. Figure \ref{fig_pf_regpath_pcr} shows regularization paths for Principal Component regression, where the penalty is the number of principal components included. In each case, the vertical line is the optimal value of the penalty chosen by 5-fold cross validation on the training data. The stocks are Coca-Cola (KO), Apple (AAPL), Exxon (XOM), Citigroup (C), Pfizer (PFE), Boeing (BA), Nike (NKE), Home Depot (HD), FedEx (FDX) and Chevron (CVX).}
\label{fig_pf_regpath}
\end{minipage}
\end{figure}

\par The nature of Ridge and Lasso shrinkage is illustrated in Figure \ref{fig_pf_regpath_ridge} and \ref{fig_pf_regpath_lasso}, respectively. Each method was applied to monthly returns for $m=10$ stocks from the S\&P500 observed for a total of $n=61$ months from 2012-2017. The figures report the estimated relative portfolio weights for varying values of the tuning parameter $\lambda$. In each case, the traditional approach estimates correspond to the values where $\lambda=0$. The estimated portfolio weights from Ridge regression follow a continuous path starting at the traditional estimates, moving towards equal weights as $\lambda$ increases. The vertical line indicates the optimal value of the penalty chosen based on 5-fold cross validation. For Lasso the shrinkage is similar, but not continuous. However for Lasso half of the assets (Citigroup, Exxon, Chevron, FedEx and Pfizer) have estimated portfolio weights equal to zero at the optimal penalty level.

\subsection{Principal Component Regression}\label{sec_pf_methods_pcr}
Principal Component regression (PCR) offers a different approach to the estimation risk problem. We may think of the full set of assets in $\X$ as a sample generated from some underlying lower dimensional data generating process, e.g. based on macroeconomic fundamentals or industry-specific factors. The idea of PCR is to base the portfolio weights on the full set of assets, but only the variation in these assets that can be attributed to the underlying factors. The low dimensional components summarise most of the variation in the return data, but is less subject to estimation risk because of the dimensionality reduction. 
 
\par It can be shown that finding a $\lambda$-dimensional subspace of the data that explains most of the original return data (i.e. maximizes the variance of the return data) is the same as obtaining the first $\lambda$ eigenvectors, or principal components, of the data. PCR then proceed by projecting each asset onto the lower dimensional space using the principal components, before using linear regression on the reduced data to obtain the asset weights.

\par PCR can be implemented as follows. Let $\mathbf{P}$ denote the $m\times m$ matrix with the principal components of $\X$ stored in each column. Let $\mathbf{P}_\lambda$ denote the $m \times \lambda$ matrix where only the first $\lambda$ principal components are included. Each asset is projected onto the lower dimensional space using $\X_\lambda=\X\mathbf{P}_\lambda$, and the low-dimensional representation of the portfolio is obtained from 
\begin{equation}
\boldsymbol{\gamma} = (\X_\lambda'\X_\lambda)^{-1}\X_\lambda'\y
\end{equation}
where $\boldsymbol{\gamma}$ is the $\lambda\times 1$ vector of principal component weights and $\y=\mathbf{1}\bar{r}$. The optimal portfolio weights are then estimated as $\hat{\boldsymbol{\theta}}_\text{P} = \mathbf{P}_\lambda \boldsymbol{\gamma}$. Similar to the regularization approaches discussed in Section \ref{sec_pf_methods_ridgelasso}, the number of principal components $\lambda$ can be chosen based on cross-validation.

\par PCR is related to (\ref{eq_pf_framework_ml_lambda}) as follows. By basing the asset positions on the top $\lambda$ principal components, it is implied that the solution will be orthogonal to the bottom $m-\lambda$ principal components. Thus, PCR uses the penalty $P(\boldsymbol{\theta}) = \mathbf{P}_{-\lambda}'\boldsymbol{\theta}=\zero$, where $\mathbf{P}_{-\lambda}$ denotes the bottom $m-\lambda$ eigenvectors. 

\begin{figure}
\centering
\begin{minipage}{1\textwidth}
\begin{subfigure}{0.50\textwidth}
\includegraphics[width=\linewidth]{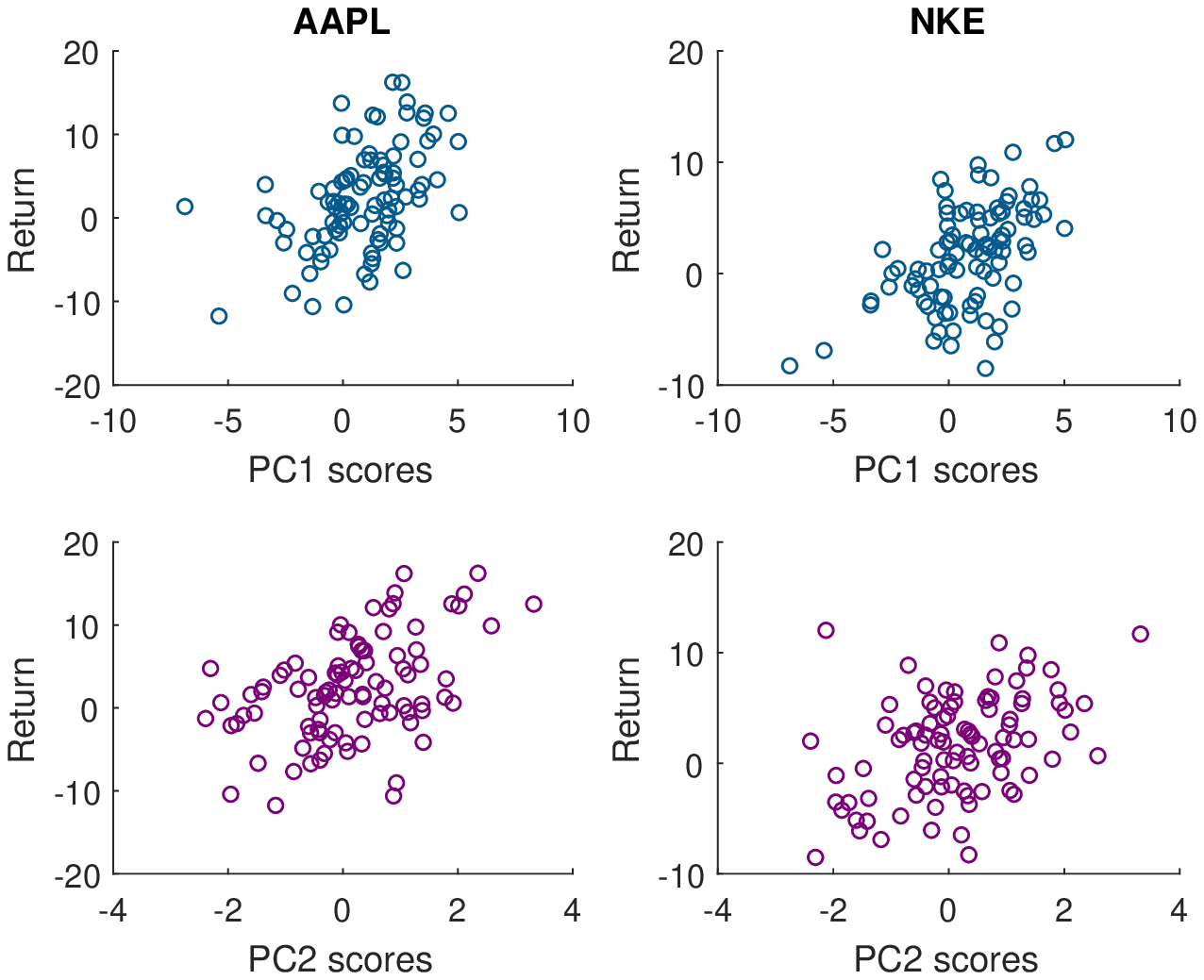}
\caption{Apple and Nike}
\label{fig_pf_pcr_1}
\end{subfigure}\hspace*{\fill}
\begin{subfigure}{0.50\textwidth}
\includegraphics[width=\linewidth]{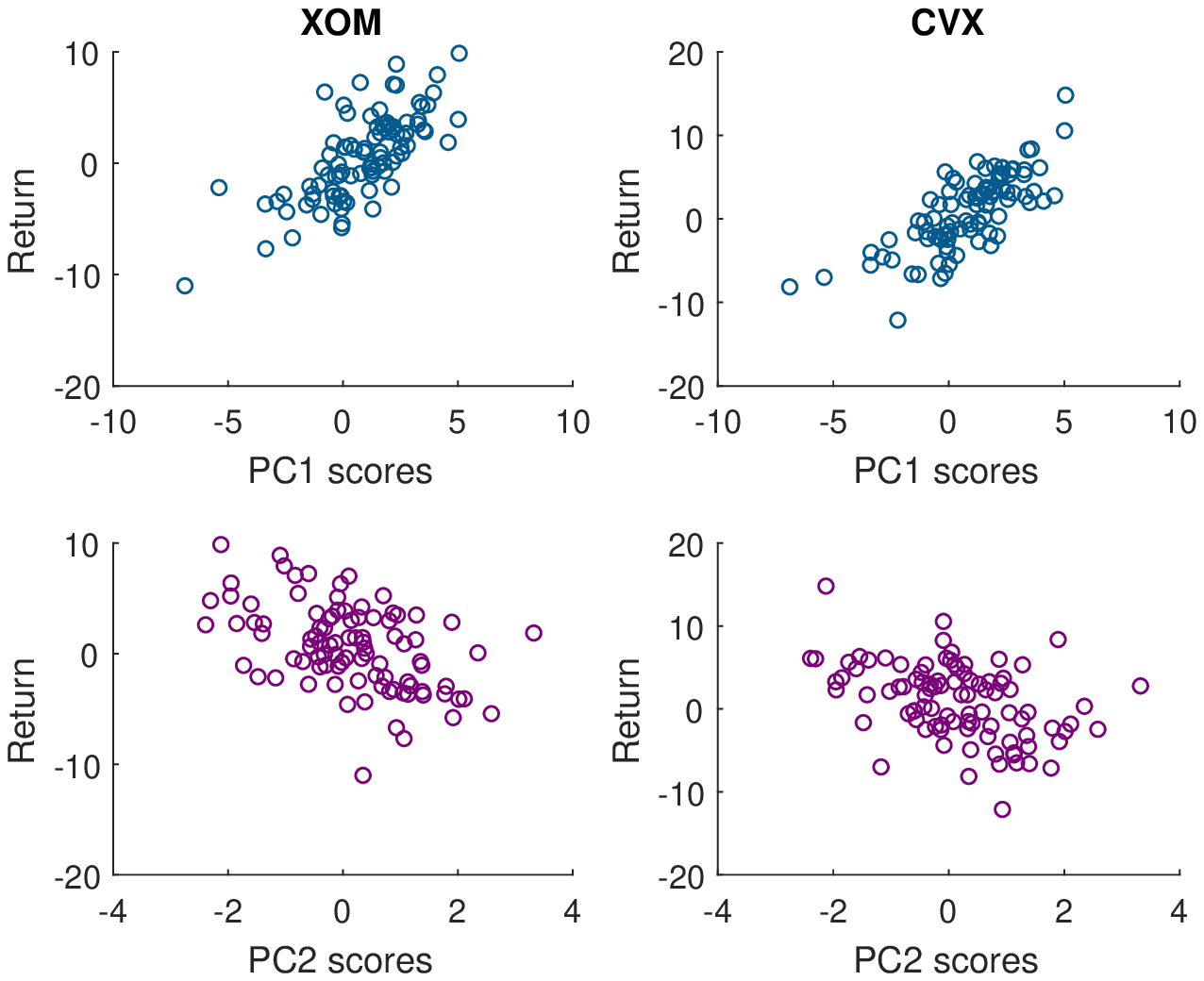}
\caption{Exxon and Chevron}
\label{fig_pf_pcr_2}
\end{subfigure}
\caption{\footnotesize \textbf{Principal component analysis.} Principal component analysis of monthly returns for 10 stocks from the S\&P500 from August 2012 to December 2017. Figure \ref{fig_pf_pcr_1} plots the first and second principal component scores (first and second column of $\X_\lambda$) versus the return data for Apple (AAPL) and Nike (NKE). Figure \ref{fig_pf_pcr_2} shows a similar plot for Exxon (XOM) and Chevron (CVX).}
\label{fig_pf_pcr}
\end{minipage}
\end{figure}

\par Continuing the S\&P500 illustration, the two first principal components of the return data are plotted against the return of a selection of stocks in Figure \ref{fig_pf_pcr}. There is a clear indication of a positive relationship between the first principal component and the return for Apple, Nike, Exxon and Chevron, while there is a negative relationship between the return on the oil companies and the second principal component. The regularization path of PCR may be derived in a similar fashion to that of Ridge and Lasso. Figure \ref{fig_pf_regpath_pcr} shows how the portfolio weights change by reducing the number of principal components from the maximum of 10 down to only a single component. Using 5-fold cross-validation, the optimal number of principal components is 2 in this case. With two principal components worth of information, the relative portfolio weights are similar for all assets, ranging from close to zero up to 0.2.

\subsection{Spike and Slab Regression}\label{sec_pf_methods_ss}
In the Bayesian context, a common way of introducing sparsity in regression models is by using a Spike and Slab prior.\footnote{The Spike and Slab method is explained in e.g. \cite{mitchell1988bayesian}, \cite{george1997approaches} and \cite{scott2014predicting}.} This approach may be viewed as a Bayesian method for approximating the best subset problem discussed in Section \ref{sec_pf_framework_sample}. The best subset idea is implemented using a $m$ dimensional vector $\boldsymbol{\eta}$, where each element $\eta_j=1$ if $\theta_j\neq 0$ and $\eta_j=0$ if $\theta_j=0$. This vector summarises which assets are included in the portfolio. The objective is to estimate the posterior distribution of $\boldsymbol{\eta}$ and the corresponding asset positions $\boldsymbol{\theta}_{\boldsymbol{\eta}}$, where the subscript indicates that the portfolio weight vector only includes the assets where $\eta_j=1$. The likelihood of the regression model is given by $\N(\X_{\boldsymbol{\eta}}\boldsymbol{\theta}_{\boldsymbol{\eta}},\phi^2\I)$ where, in addition to the parameters $\boldsymbol{\theta}$ and $\phi^2$, the asset inclusion vector $\boldsymbol{\eta}$ is unknown. The Spike and Slab prior is
\begin{equation}
p(\boldsymbol{\theta},\phi^2,\boldsymbol{\eta}) = p(\boldsymbol{\theta}_{\boldsymbol{\eta}}|\phi^2,\boldsymbol{\eta})p(\phi^2|\boldsymbol{\eta})p(\boldsymbol{\eta})
\end{equation}
where $p(.)$ is used as a generic notation for densities. The prior for the asset inclusion vector $p(\boldsymbol{\eta})$ is the ``Spike'' and $p(\boldsymbol{\theta}_{\boldsymbol{\eta}}|\phi^2,\boldsymbol{\eta})$ is the ``Slab''. It is common to assume a Bernoulli distribution for the Spike
\begin{equation}\label{eq_pf_methods_ss_prior_eta}
\boldsymbol{\eta} \sim \prod_{j=1}^m \pi_j^{\eta_j}(1-\pi_j)^{1-\eta_j}
\end{equation} 
where $\pi_j$ denotes the inclusion probability of asset $j$. The uniform prior assumes that each asset has an equal chance of being included or excluded, i.e. $\pi_j=1/2$ for each $j$. In some settings it could be useful to alter this specification, by specifically excluding or including assets by setting $\pi_j=0$ or $\pi_j=1$, respectively. Another alternative is to specify an expected portfolio size, $k$, and then take $\pi_j=k/m$. I will use the uniform prior in this paper. The remaining priors are specified as
\begin{eqnarray}
\boldsymbol{\theta}_{\boldsymbol{\eta}}|\phi^2,\boldsymbol{\eta} & \sim & \N(\boldsymbol{\theta}_{\boldsymbol{\eta}}^0,\phi^2 \mathbf{V}_{\boldsymbol{\eta}}^0)\label{eq_pf_methods_ss_prior_theta}\\
\phi^2|\boldsymbol{\eta} & \sim & \mathcal{G}^{-1}(a^0,b^0)\label{eq_pf_methods_ss_prior_sigma}
\end{eqnarray}
where $\N$ is the normal distribution with prior mean $\boldsymbol{\theta}_{\boldsymbol{\eta}}^0$ and variance $\phi^2\mathbf{V}_{\boldsymbol{\eta}}^0$, and $\mathcal{G}^{-1}$ is the inverse Gamma distribution with shape $a^0$ and scale $b^0$. The prior is illustrated in Figure \ref{fig_pf_ss_prior}, where the Spike ensures probability mass at zero, while the Slab distributes probability mass to a broad set of values for the portfolio weight. Using the likelihood, (\ref{eq_pf_methods_ss_prior_eta}), (\ref{eq_pf_methods_ss_prior_theta}) and (\ref{eq_pf_methods_ss_prior_sigma}) it can be shown, see e.g. \cite{murphy2012machine}, that the posteriors of $\boldsymbol{\theta}$ and $\phi^2$ conditional on the included assets $\boldsymbol{\eta}$ are given by
\begin{equation}\label{eq_pf_methods_ss_post_theta_sigma}
\boldsymbol{\theta}_{\boldsymbol{\eta}}|\phi^2,\boldsymbol{\eta},\y \sim \N(\boldsymbol{\theta}_{\boldsymbol{\eta}}^1,\phi^2\mathbf{V}_{\boldsymbol{\eta}}^1) \hspace{1cm} \text{and} \hspace{1cm}\phi^2|\boldsymbol{\eta},\y \sim \mathcal{G}^{-1}(a^1,b^1) 
\end{equation}
where the posterior asset mean is $\boldsymbol{\theta}_{\boldsymbol{\eta}}^1 = \mathbf{V}_{\boldsymbol{\eta}}^1(\X_{\boldsymbol{\eta}}'\y + (\mathbf{V}_{\boldsymbol{\eta}}^0)^{-1} \boldsymbol{\theta}_{\boldsymbol{\eta}}^0)$ and the variance is $\mathbf{V}_{\boldsymbol{\eta}}^1 = (\X_{\boldsymbol{\eta}}'\X_{\boldsymbol{\eta}} + (\mathbf{V}_{\boldsymbol{\eta}}^0)^{-1})^{-1}$. The posterior shape is $a^1 = a^0 + \frac{1}{2}n$ and the scale is $b^1 = b^0 + \frac{1}{2}(\y'\y + \boldsymbol{\theta}_{\boldsymbol{\eta}}^{0'} (\mathbf{V}_{\boldsymbol{\eta}}^0)^{-1}\boldsymbol{\theta}_{\boldsymbol{\eta}}^0 - \boldsymbol{\theta}_{\boldsymbol{\eta}}^{1'}(\mathbf{V}_{\boldsymbol{\eta}}^1)^{-1}\boldsymbol{\theta}_{\boldsymbol{\eta}}^1)$. Given the above, integrating out $\boldsymbol{\theta}_{\boldsymbol{\eta}}$ and $\phi^2$ gives a closed form expression for the posterior of the inclusion variable
\begin{equation}\label{eq_pf_methods_ss_post_eta}
p(\boldsymbol{\eta}|\y) \propto \frac{1}{(2\pi)^{n/2}} \frac{|\mathbf{V}_{\boldsymbol{\eta}}^1|^{1/2}}{|\mathbf{V}_{\boldsymbol{\eta}}^0|^{1/2}} \frac{\Gamma(a^1)}{\Gamma(a^0)} \frac{(b^0)^{a^0}}{(b^1)^{a^1}} \prod_{j=1}^m \pi_j^{\eta_j} (1-\pi_j)^{1-\eta_j}
\end{equation}
Using the traditional approach estimate for the included assets, $\hat{\boldsymbol{\theta}}_{\boldsymbol{\eta}}=(\X_{\boldsymbol{\eta}}'\X_{\boldsymbol{\eta}})^{-1}\X_{\boldsymbol{\eta}}'\y$ and Zellner's g-prior $(\mathbf{V}_{\boldsymbol{\eta}}^0)^{-1}=\frac{g}{n}\X_{\boldsymbol{\eta}}'\X_{\boldsymbol{\eta}}$ it follows directly that
\begin{equation}\label{eq_pf_methods_existing_empbayes}
\boldsymbol{\theta}_{\boldsymbol{\eta}}^1 = \frac{n}{n+g}\hat{\boldsymbol{\theta}}_{\boldsymbol{\eta}} + \frac{g}{n+g}\boldsymbol{\theta}_{\boldsymbol{\eta}}^0
\end{equation}
The posterior mean is therefore a weighted combination of the traditional estimate and the prior, conditional on included assets only.
\begin{algorithm}[t]
\caption{Spike and Slab portfolio selection}
\begin{algorithmic}[1]
\item Set a starting value $\boldsymbol{\eta}^{(0)}=(1,\hdots,1)$
\item For Gibbs sampling iteration $i=1,\hdots,N$
\subitem Set $\boldsymbol{\eta}^{(i)} \leftarrow \boldsymbol{\eta}^{(i-1)}$.
\subitem For each $j=1,\hdots,m$ in random order
\subsubitem Set $\eta_j^{(i)}=1$ if $u<h(\eta_j=1)/(h(\eta_j=1)+h(\eta_j=0))$
\subsubitem where $h(\eta_j=x) = p(\eta_j=x,\boldsymbol{\eta}_{-j}^{(i)}|\y)$ and $u\sim \text{uniform}(0,1)$.
\subitem Draw $(\phi^2)^{(i)}$ from $p(\phi^2|\boldsymbol{\eta}^{(i)},\y)$.
\subitem Draw $\boldsymbol{\theta}_{\boldsymbol{\eta}}^{(i)}$ from $p(\boldsymbol{\theta}_{\boldsymbol{\eta}}|(\phi^2)^{(i)},\boldsymbol{\eta}^{(i)},\y)$.
\end{algorithmic}
\label{tab_pf_methods_ss}
\end{algorithm}

\par The Spike and Slab regression is implemented using Gibbs-sampling as outlined in Algorithm \ref{tab_pf_methods_ss}. As a starting point, all assets are assumed included in the portfolio. From this starting value, (\ref{eq_pf_methods_ss_post_eta}) can be used to obtain the posterior of the included assets. For each asset $j$, switching between inclusion ($\eta_j=1$) and exclusion ($\eta_j=0$), gives two evaluations of the posterior inclusion density that can be used to compute an asset inclusion probability. Drawing a uniform random number between zero and one then determines if the asset should be included or excluded, and the inclusion vector $\boldsymbol{\eta}$ is updated accordingly. Once all assets have been examined in random order, the posterior portfolio vector and the corresponding noise level can be drawn using (\ref{eq_pf_methods_ss_post_theta_sigma}). Repeating this procedure several thousand times gives estimated posterior distributions of the optimal portfolio weights and the inclusion probability of each asset.

\par The results from applying Spike and Slab regression to the illustrative S\&P example are reported in Figure \ref{fig_pf_ss_inclusion} and Figure \ref{fig_pf_ss_est}. The method suggests enforcing sparsity by excluding all assets except Home Depot and Coca Cola. Home Depot is included in close to 100\% of the Gibbs sampling iterations, while Coca Cola is the second most frequent included asset, with an inclusion probability below 10\%. The results are broadly in line with the regularization approaches discussed in Section \ref{sec_pf_methods_ridgelasso}. 

\begin{figure}
\centering
\begin{minipage}{1\textwidth}
\begin{subfigure}{0.33\textwidth}
\includegraphics[width=\linewidth]{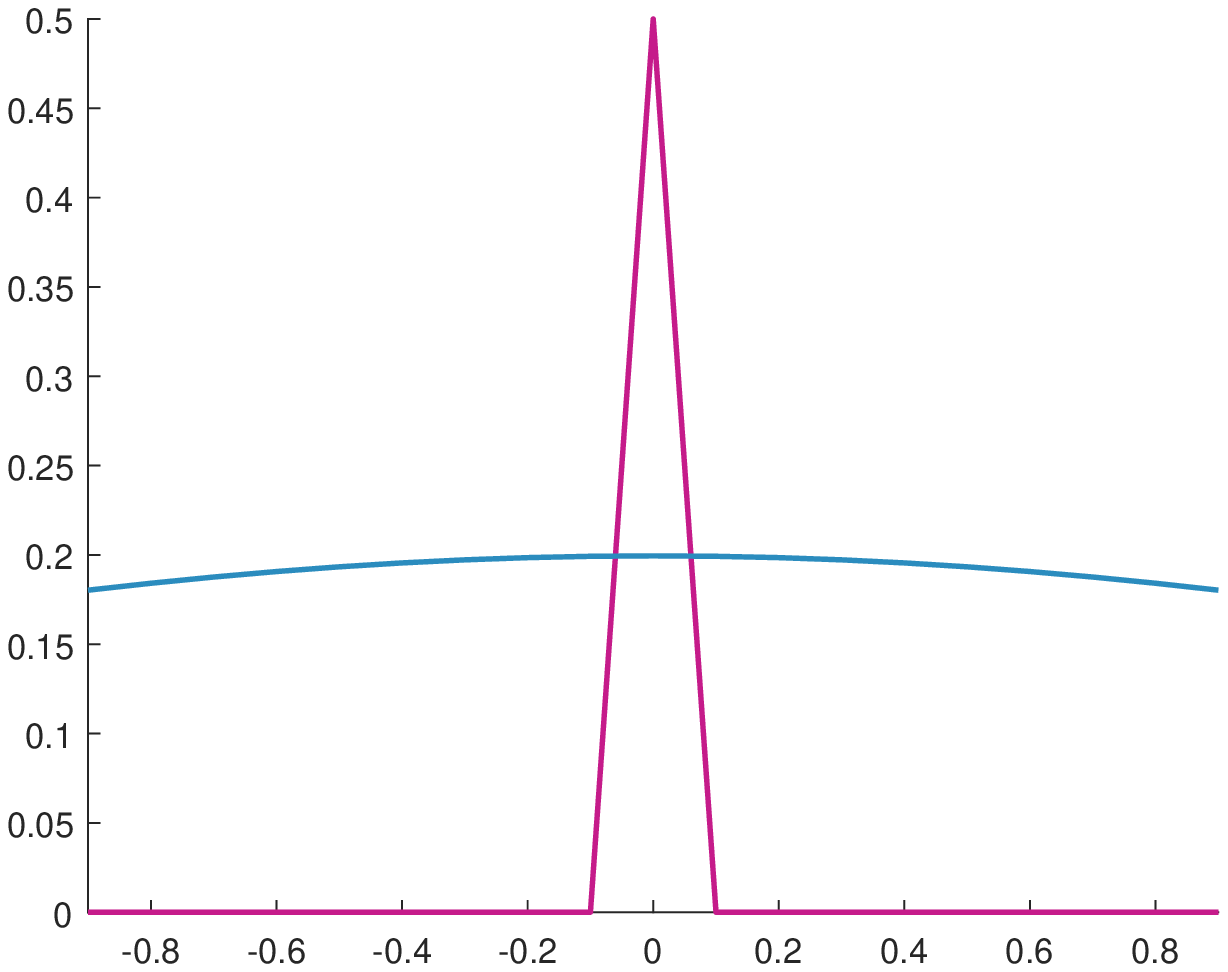}
\caption{Spike and slab}
\label{fig_pf_ss_prior}
\end{subfigure}\hspace*{\fill}
\begin{subfigure}{0.33\textwidth}
\includegraphics[width=\linewidth]{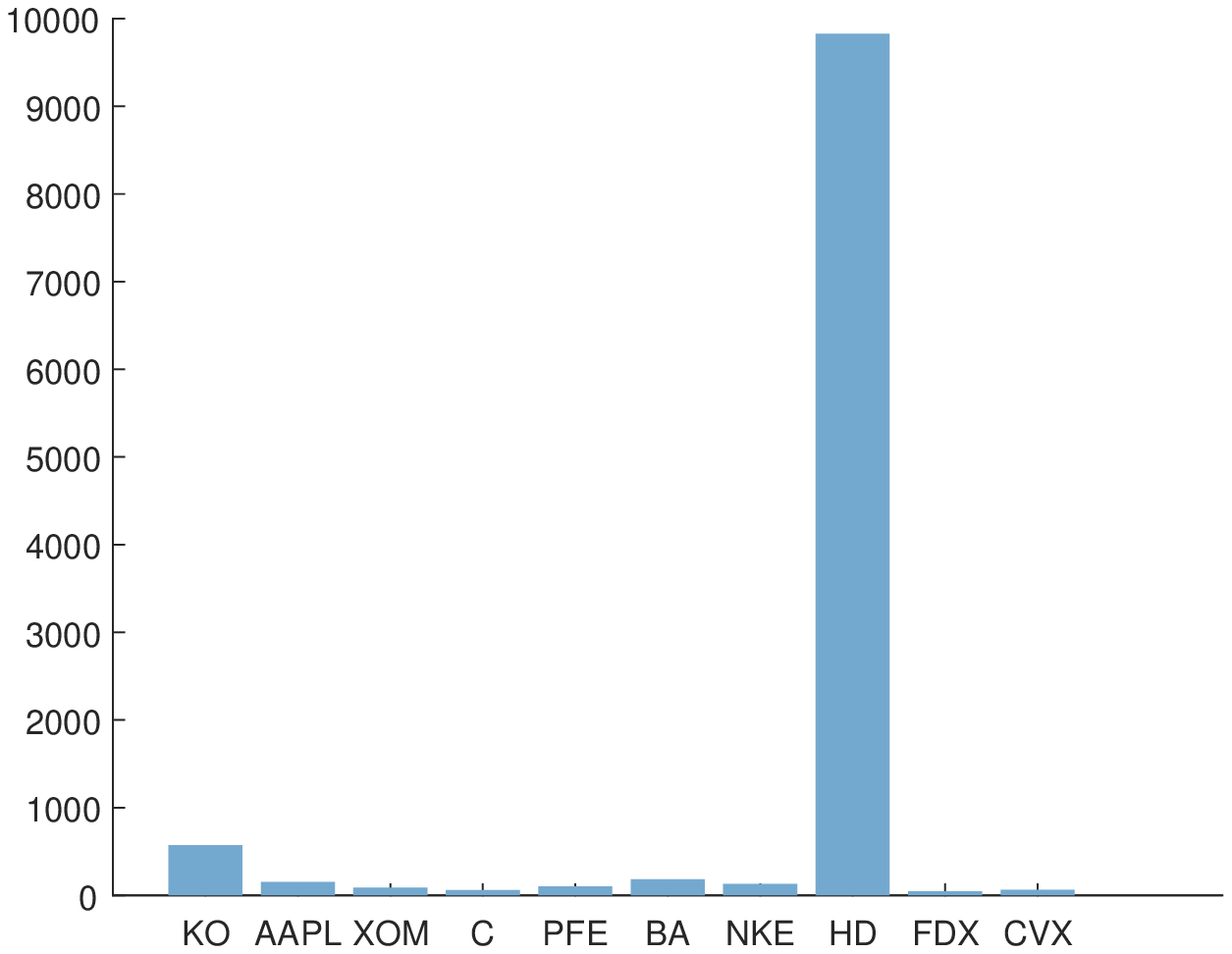}
\caption{Inclusion frequencies}
\label{fig_pf_ss_inclusion}
\end{subfigure}\hspace*{\fill}
\begin{subfigure}{0.33\textwidth}
\includegraphics[width=\linewidth]{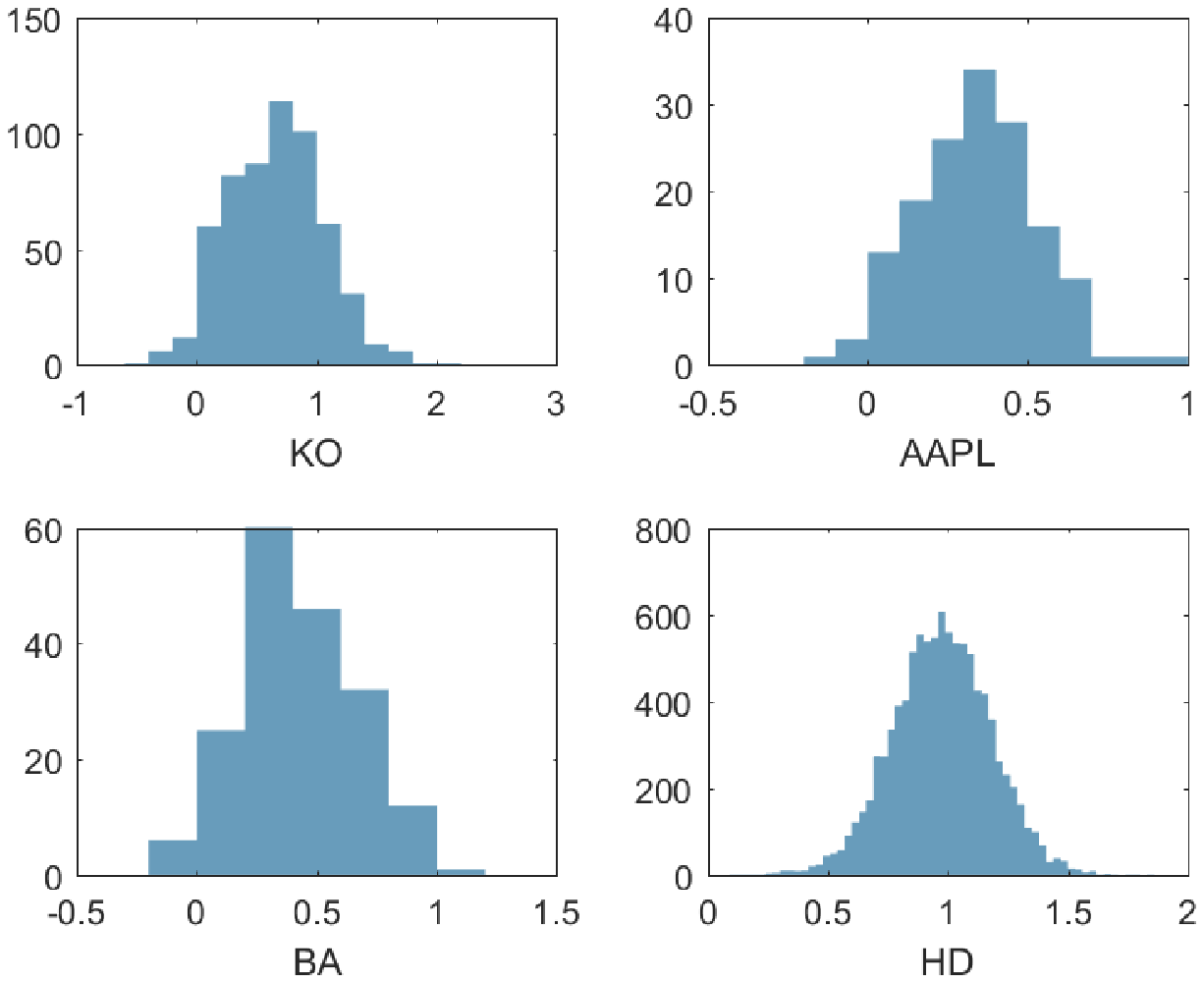}
\caption{Weight estimates}
\label{fig_pf_ss_est}
\end{subfigure}
\caption{\footnotesize \textbf{Spike and Slab portfolio selection.} Based on monthly returns for 10 stocks from the S\&P500 from August 2012 to December 2017. Figure \ref{fig_pf_ss_prior}: Illustrative Spike and Slab prior based on Bernoulli and Gaussian distributions, respectively. Figure \ref{fig_pf_ss_inclusion}-\ref{fig_pf_ss_est}: Results based on 10000 Gibbs sampling iterations after a burn in of 5000 iterations. Figure \ref{fig_pf_ss_inclusion}: Home Depot (HD) is included in almost 100\% of the iterations, with Coca Cola (KO) being the second most frequent asset, included in less than 10\% of the iterations. Figure \ref{fig_pf_ss_est}: Histograms of the portfolio weight estimates for the four most frequent assets. Priors: For the Slab, the mean asset vector was set to $\boldsymbol{\theta}_{\boldsymbol{\eta}}^0 = \zero$ with Zellner's g-prior for the covariance, $(\mathbf{V}_{\boldsymbol{\eta}}^0)^{-1}=\frac{g}{n}\X_{\boldsymbol{\eta}}'\X_{\boldsymbol{\eta}}$ with $g=1$. The shape and scale parameters were set to $a^0=b^0=0.1$, and it was assumed an uninformative prior for the inclusion of assets, $\pi_j=0.5$ for all $j$.}
\label{fig_pf_ss}
\end{minipage}
\end{figure}

\section{Simulation Study}\label{sec_pf_sim}
\par The purpose of this simulation study is to examine the performance of ML methods in reducing estimation risk compared to the traditional approach and various benchmark approaches. The use of simulated data is beneficial for studying the generalisation error, estimation risk and bias-variance tradeoff of the various methods, as these measures depend on repeated draws of data. I use simulated data calibrated to the US stock market in order to capture returns, variances and covariances that are similar to actual US data.  

\par I randomly choose a subset of assets from the Standard \& Poor's 500 Index (S\&P500) for the period 2012-2017, and compute monthly excess returns for each asset. I define the mean and covariance of these returns as population values $\boldsymbol{\mu}$ and $\mathbf{\Sigma}$, respectively. Using this synthetic S\&P population, I draw a dataset of returns and estimate the optimal weights using all strategies presented in Section \ref{sec_pf_methods}. To assess estimation risk, each strategy is estimated on repeated draws of the training data and evaluated out of sample at the population values. The above procedure is carried out for a varying number of assets and sample sizes.

\par To be specific, given a number of assets $m$ and a number of returns $n$, I draw $K$ training datasets $\X_k\sim\N(\boldsymbol{\mu},\boldsymbol{\Sigma})$ for $k=1,\hdots,K$, where each $\X_k$ is a $n\times m$ matrix of excess returns. The estimated optimal portfolio weights for strategy $q$ and data $k$ is denoted by $\hat{\boldsymbol{\theta}}_{qk}$. The estimation risk (\ref{eq_pf_framework_estrisk}) of strategy $q$ is approximated by
\begin{equation}\label{eq_pf_sim_risk}
\hat{R}_q = (\boldsymbol{\theta}^*-\bar{\boldsymbol{\theta}}_q)'\mathbf{A}(\boldsymbol{\theta}^*-\bar{\boldsymbol{\theta}}_q) + \text{tr}\left(\mathbf{A}\mathbf{S}_q\right)
\end{equation}
where $\bar{\boldsymbol{\theta}}_q = \frac{1}{K}\sum_{k=1}^K \hat{\boldsymbol{\theta}}_{qk}$ and $\mathbf{S}_q=\frac{1}{K-1} \sum_{k=1}^K (\hat{\boldsymbol{\theta}}_{qk}-\bar{\boldsymbol{\theta}}_q)(\hat{\boldsymbol{\theta}}_{qk}-\bar{\boldsymbol{\theta}}_q)'$. The above formulation enables a study of the bias and variance of the estimated weights under strategy $q$ in repeated samples. Additionally, I compare methods using the average out of sample Sharpe ratio across datasets
\begin{equation}\label{eq_pf_sim_sharpe}
\hat{s}_q = \frac{1}{K}\sum_{k=1}^K \frac{\boldsymbol{\mu}'\hat{\boldsymbol{\theta}}_{qk}}{\sqrt{\hat{\boldsymbol{\theta}}_{qk}'\boldsymbol{\Sigma}\hat{\boldsymbol{\theta}}_{qk}}}
\end{equation}
Each strategy was implemented as follows. The optimal population solution (\ref{eq_pf_framework_thetaPop}) was computed directly using the population moments. Using $K=100$ datasets of returns from the population, the traditional approach, Ridge, Lasso, PCR and Spike and Slab was applied to each dataset. For Ridge, Lasso and PCR, 5-fold cross validation was used to choose the penalty parameter $\lambda$. Spike and Slab was implemented using a zero mean prior $\boldsymbol{\theta}_{\boldsymbol{\eta}}^0=\zero$ and Zellner's g-prior for the covariance $(\mathbf{V}_{\boldsymbol{\eta}}^0)^{-1} = \frac{q}{n}(\X_{\boldsymbol{\eta}}'\X_{\boldsymbol{\eta}})$ with $g=1$. Furthermore, I use $a^0=b^0=0.1$ and the uninformative asset prior $\pi_j=0.5$ for all $j$. As benchmarks, I use the equal weight strategy, the optimal mean-variance portfolio with a short sale restriction, the minimum variance portfolio and Empirical Bayes proposed by \cite{jorion1986bayes}.

\par The Sharpe ratio for the traditional strategy, the ML strategies and the benchmarks are reported for varying sample and portfolio sizes in Table \ref{tab_pf_sim_sharpe}. First, note that the Sharpe ratios obtained from the traditional approach are substantially lower than the population values for most cases up to 120 months. In particular, using the traditional approach for the portfolio of 10 assets gives a Sharpe ratio as low as 0.25 based on 20 months, substantially below the population value 0.62. In the case of 50 assets and 120 months, the Sharpe ratio is 1.31 compared to 2.05. Furthermore, in all cases where the number of assets is larger than the number of observed returns, the traditional strategy cannot be used due to a degenerate covariance matrix. These observations document the already well established result that the traditional approach is inefficient. As expected, the traditional approach converges to the population Sharpe ratio as the number of observations increases.

\par The second observation is that all ML algorithms yield similar results for the Sharpe ratio, well above the traditional approach up to 60 months, and similar to the traditional approach from 120 months onwards. Even in highly degenerate cases with 50 assets and 20 returns, the ML strategies all yield reasonably high Sharpe ratios. The reason is that in small samples, the mean and covariance of the assets will be highly unstable across folds, thus making it optimal to set high penalty values. This leads to less variability in the weights and thus also less variability in the portfolio mean and standard deviation across datasets. As the number of observations gets large, the sample mean and covariance of asset returns will be similar to the population moments in all subsamples. In that case the ML strategies find it optimal to set the penalty level low or to zero. Thus, all ML strategies relying on cross-validation will converge to the traditional approach as the number of observations increases. 

\begin{figure}
\centering
\begin{minipage}{1\textwidth}
\begin{subfigure}{0.32\textwidth}
\includegraphics[width=\linewidth]{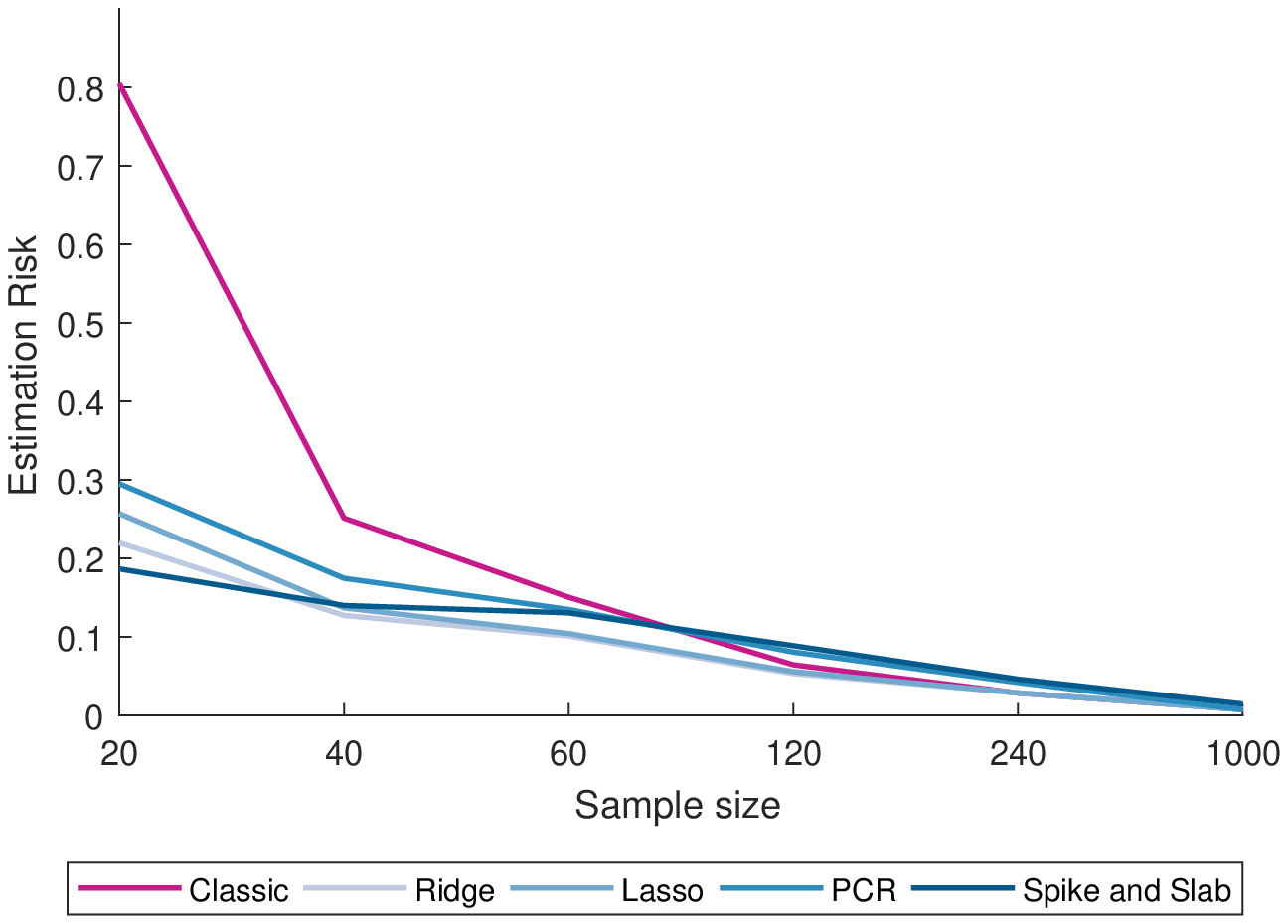}
\caption{Estimation risk}
\label{fig_pf_sim_risk_full}
\end{subfigure}\hspace*{\fill}
\begin{subfigure}{0.32\textwidth}
\includegraphics[width=\linewidth]{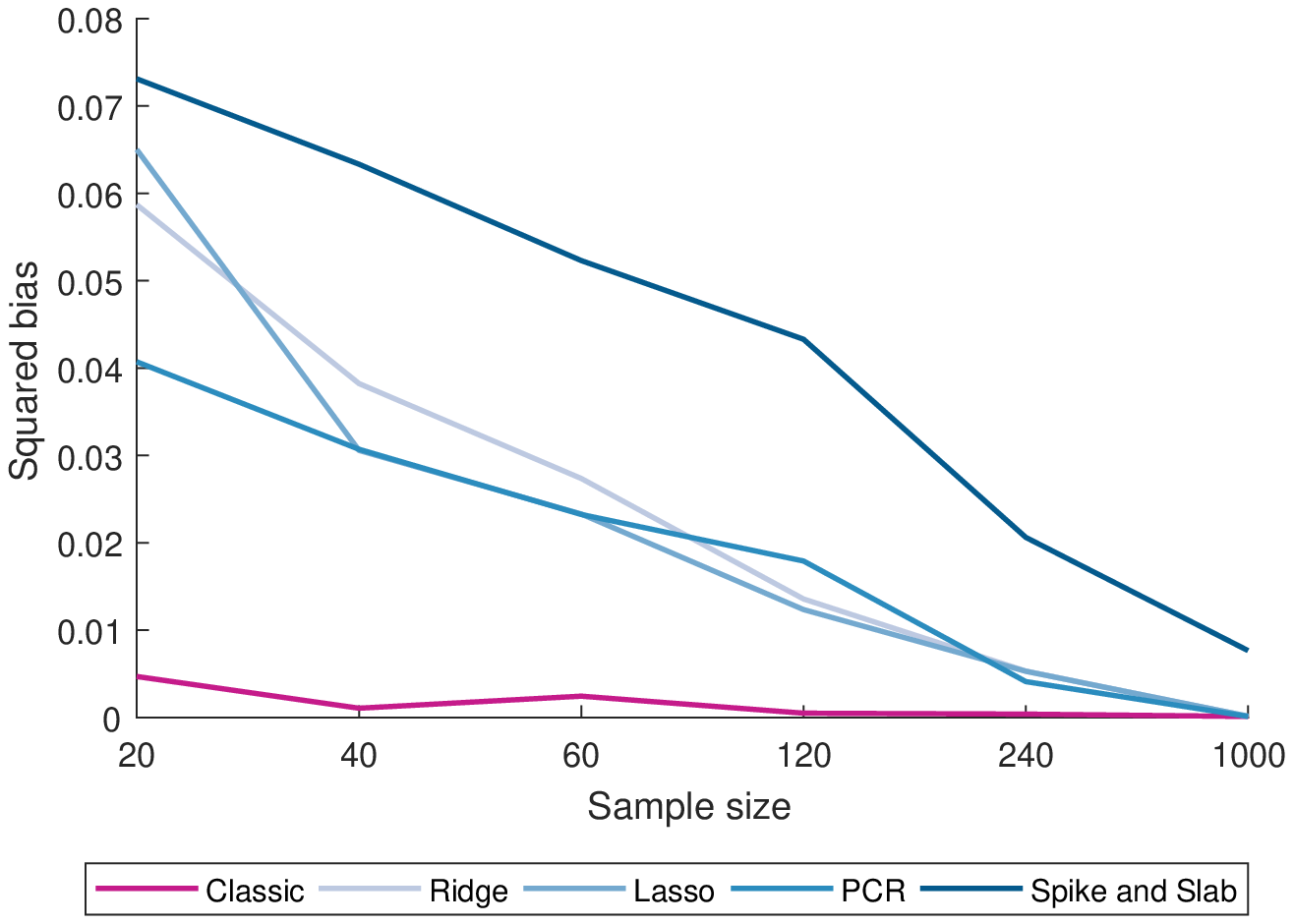}
\caption{Squared bias}
\label{fig_pf_sim_risk_bias}
\end{subfigure}\hspace*{\fill}
\medskip
\begin{subfigure}{0.32\textwidth}
\includegraphics[width=\linewidth]{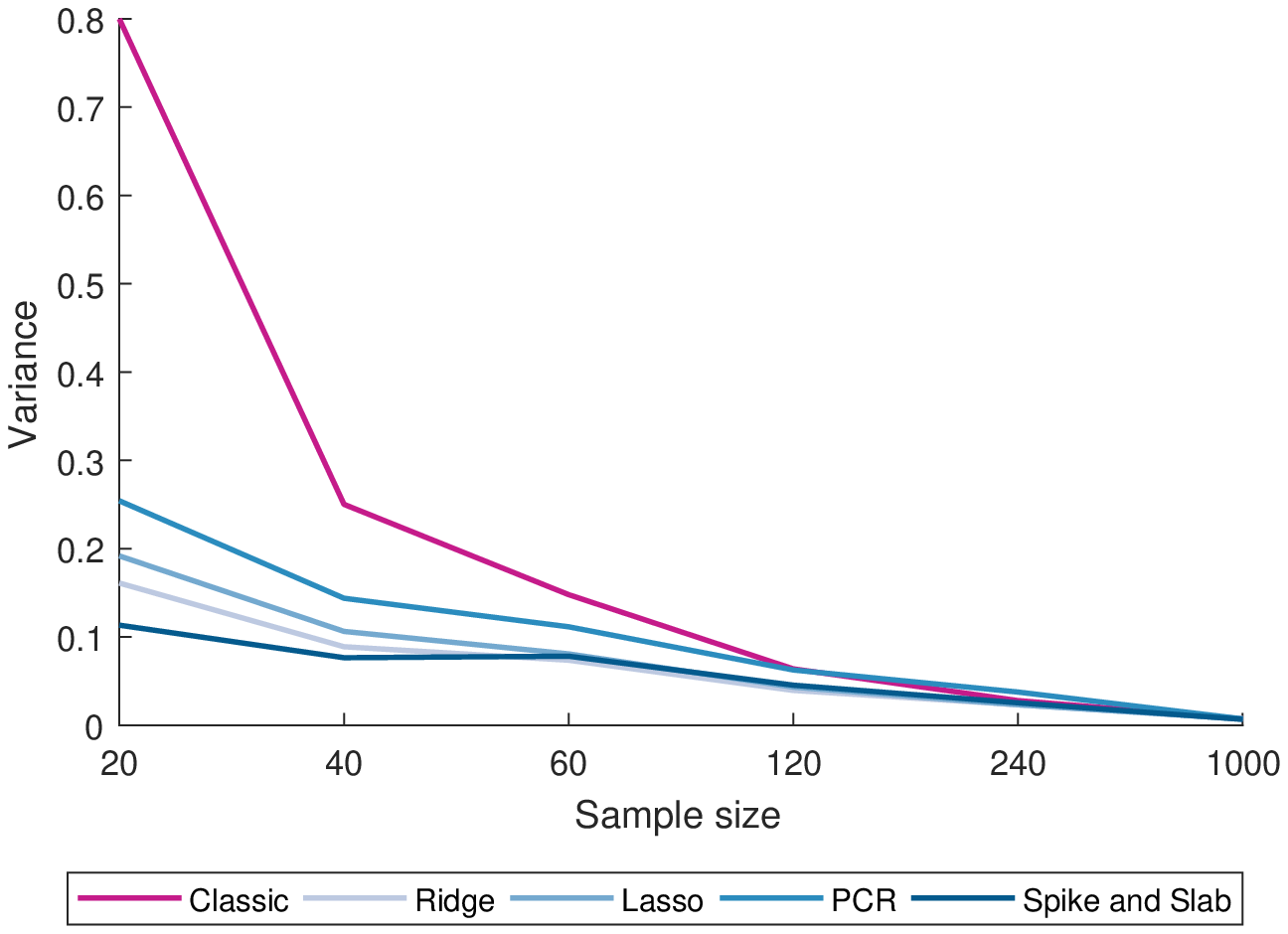}
\caption{Variance}
\label{fig_pf_sim_risk_var}
\end{subfigure}\\
\caption{\footnotesize \textbf{Decomposition of estimation risk.} Plot of average estimation risk (\ref{eq_pf_sim_risk}) for varying sample sizes for the case of $m=10$. Figure \ref{fig_pf_sim_risk_bias} and \ref{fig_pf_sim_risk_var} provide the bias-variance decomposition of the estimation risk.}
\label{fig_pf_sim_risk}
\end{minipage}
\end{figure}

\par Third, the ML strategies outperform the benchmark strategies in most cases. The equal weight strategy works well for short sample sizes and small portfolios, slightly outperforming some of the ML algorithms. However, the problem of equal weighting is unmasked when the number of observations increases because the data is increasingly containing information about the population moments, and thus also about the optimal weights. Similarly, the mean-variance portfolio without short selling can be very effective for small sample sizes, as the non-negativity restriction provides a lower bound for the weights. However, as the number of observations increases, the sample may contain precise information about negative optimal weights. In such cases, the non-negativity restriction could be harmful. The minimum variance portfolio imposes the constraint that means are irrelevant for portfolio choice. A large literature has documented that this strategy performs well in practice due to the difficulty in estimating means from the data. My simulation study suggests that it performs adequately for a moderate number of assets ($m=10$) and low sample sizes. However, as the number of observations increases, disregarding information about the means leads to relatively poor Sharpe ratios. In general, the minimum variance strategy will only work for non-degenerate cases because it relies on estimation of the covariance matrix. 

\par The three benchmark approaches discussed above have the same general problem of imposing ``hard'' constraints that are invariant to changes in sample size. As the number of observations increases, none of these strategies will approach the population Sharpe ratio. In contrast, the ML approaches are more flexible, imposing a large penalty in cases with severe estimation risk and a low or zero penalty when estimation risk is low.

\par The last benchmark, the Empirical Bayes, may be viewed as a ML strategy due to the data-driven selection of the penalty parameter. Indeed, the results suggest that the strategy is adapting to the traditional estimates as the number of observations increases. However, each weight is shrunk towards the minimum variance portfolio, which is degenerate when $m>n$.

\par Table \ref{tab_pf_sim_risk} reports that the estimation risk is significantly lower for all ML algorithms compared to the traditional approach. An intuitive explanation for this result is provided in Figure \ref{fig_pf_sim_risk}, where the estimation risk is decomposed into squared bias and variance for the case of $m=10$. Note that the traditional approach almost achieves zero bias even in small samples, but bears too much variance. The ML approaches reduce variance by accepting bias in the portfolio weights, in total leading to lower estimation risk.

\begin{landscape}
\begin{table}
\scriptsize
\centering
\begin{tabular}{p{6cm} rrrrrr|rrrrrr|rrrrrr}
\midrule
Assets ($m$) & \multicolumn{6}{c}{10} & \multicolumn{6}{c}{25} & \multicolumn{6}{c}{50}\\
Sample size	($n$) & 20 & 40 & 60 & 120 & 240 & 1000 & 20 & 40 & 60 & 120 & 240 & 1000 & 20 & 40 & 60 & 120 & 240 & 1000\\
\midrule
Population 		& 0.62 & 0.62 & 0.62 & 0.62 & 0.62 & 0.62
				& 1.05 & 1.05 & 1.05 & 1.05 & 1.05 & 1.05 
				& 2.05 & 2.05 & 2.05 & 2.05 & 2.05 & 2.05\\
MV			 	& 0.25 & 0.41 & 0.47 & 0.55 & 0.59 & 0.61
				& - & 0.46 & 0.67 & 0.86 & 0.96 & 1.02
				& - & - & 0.38 & 1.31 & 1.79 & 1.99\\ \\

\textbf{Machine Learning}\\
Ridge			& 0.39 & 0.47 & 0.50 & 0.55 & 0.58 & 0.61
				& 0.35 & 0.67 & 0.74 & 0.86 & 0.96 & 1.03
				& 0.51 & 0.60 & 0.86 & 1.50 & 1.78 & 1.99\\
Lasso			& 0.37 & 0.46 & 0.50 & 0.55 & 0.58 & 0.61
				& 0.46 & 0.63 & 0.73 & 0.87 & 0.95 & 1.03
				& 0.42 & 0.67 & 0.84 & 1.43 & 1.79 & 1.99\\
PCR				& 0.36 & 0.42 & 0.47 & 0.52 & 0.57 & 0.61
				& 0.43 & 0.59 & 0.65 & 0.81 & 0.95 & 1.03
				& 0.41 & 0.59 & 0.69 & 1.28 & 1.79 & 1.99\\
Spike \& Slab	& 0.39 & 0.44 & 0.46 & 0.50 & 0.56 & 0.60
				& 0.33 & 0.63 & 0.67 & 0.77 & 0.85 & 0.98
				& 0.39 & 0.34 & 0.75 & 0.86 & 1.30 & 1.95\\ \\

\textbf{Benchmarks}\\
Equal weights	& 0.40 & 0.40 & 0.40 & 0.40 & 0.40 & 0.40
				& 0.38 & 0.38 & 0.38 & 0.38 & 0.38 & 0.38
				& 0.33 & 0.33 & 0.33 & 0.33 & 0.33 & 0.33\\
MV-C			& 0.45 & 0.47 & 0.48 & 0.49 & 0.50 & 0.51
				& 0.55 & 0.59 & 0.63 & 0.67 & 0.68 & 0.69
				& 0.56 & 0.64 & 0.68 & 0.72 & 0.76 & 0.78\\
Min Variance	&  0.31 & 0.34 & 0.38 & 0.38 & 0.39 & 0.40
				& - & 0.31 & 0.38 & 0.44 & 0.48 & 0.50
				& - & - & 0.17 & 0.31 & 0.38 & 0.41\\
EB				& 0.27 & 0.40 & 0.46 & 0.54 & 0.59 & 0.61
				& - & 0.42 & 0.64 & 0.86 & 0.96 & 1.03
				& - & - & 0.34 & 1.31 & 1.79 & 1.99\\
\bottomrule
\end{tabular}
\caption{\footnotesize{\textbf{Sharpe ratios.} Average Sharpe ratios computed using (\ref{eq_pf_sim_sharpe}) based on $K=100$ repeated draws of the data $\X_k\sim\N(\boldsymbol{\mu},\boldsymbol{\Sigma})$ for $k=1,\hdots,K$ for varying sample sizes $n$ and number of assets $m$. The population values $\boldsymbol{\mu}$ and $\boldsymbol{\Sigma}$ are computed using a random draw of $m$ stocks from the S\&P500. MV denotes the traditional approach, MV-C the traditional approach without short-selling and EB the Empirical Bayes approach.}}
\label{tab_pf_sim_sharpe}
\end{table}
\begin{table}
\scriptsize
\centering
\begin{tabular}{p{6cm} rrrrrr|rrrrrr|rrrrrr}
\midrule
Assets ($m$) & \multicolumn{6}{c}{10} & \multicolumn{6}{c}{25} & \multicolumn{6}{c}{50}\\
Sample size	($n$) & 20 & 40 & 60 & 120 & 240 & 1000 & 20 & 40 & 60 & 120 & 240 & 1000 & 20 & 40 & 60 & 120 & 240 & 1000\\
\midrule
MV				& 0.81 & 0.25 & 0.15 & 0.06 & 0.02 & 0.01
				& -	   & 0.84 & 0.37 & 0.12 & 0.04 & 0.01 
				& -    & -    & 1.15 & 0.14 & 0.05 & 0.01 \\
Ridge			& 0.22 & 0.13 & 0.10 & 0.05 & 0.03 & 0.01
				& 1.02 & 0.25 & 0.19 & 0.11 & 0.05 & 0.01
				& 0.61 & 0.90 & 0.42 & 0.13 & 0.05 & 0.01 \\
Lasso			& 0.26 & 0.14 & 0.10 & 0.06 & 0.03 & 0.01 
				& 0.47 & 0.28 & 0.21 & 0.11 & 0.05 & 0.01
				& 0.73 & 0.58 & 0.44 & 0.13 & 0.05 & 0.01 \\
PCR				& 0.30 & 0.18 & 0.14 & 0.08 & 0.04 & 0.01
				& 0.49 & 0.36 & 0.30 & 0.16 & 0.06 & 0.01
				& 0.75 & 0.65 & 0.58 & 0.16 & 0.05 & 0.01 \\
Spike \& Slab	& 0.19 & 0.14 & 0.13 & 0.09 & 0.05 & 0.01
				& 2.17 & 0.27 & 0.23 & 0.16 & 0.11 & 0.04
				& 1.14 & 2.23 & 0.47 & 0.39 & 0.19 & 0.02 \\
\bottomrule
\end{tabular}
\caption{\footnotesize{\textbf{Estimation risk.} The table reports estimation risk based on (\ref{eq_pf_sim_risk}) using $K=100$ samples of the data $\X_k\sim\N(\boldsymbol{\mu},\boldsymbol{\Sigma})$ for $k=1,\hdots,K$, for varying number of sample sizes $n$ and number of assets $m$. The population values $\boldsymbol{\mu}$ and $\boldsymbol{\Sigma}$ are computed using a random draw of $m$ stocks from the S\&P500. MV denotes the traditional approach.}}
\label{tab_pf_sim_risk}
\end{table}
\end{landscape}

\section{Estimation Risk in Applications}\label{sec_pf_empirical}

\begin{table}
\scriptsize
\centering
\begin{tabular}{lrrrrrr}
\midrule
Dataset	 & Abbreviation & Period & $m$ & $T$ & $n$ & Source\\
\midrule
Standard \& Poor's  & S\&P-20 & Jan90-Oct17 & 20 & 334 & 120 & Kaggle.com \\
Standard \& Poor's  & S\&P-50 & Jan90-Oct17 & 50 & 334 & 120 & Kaggle.com \\
Standard \& Poor's  & S\&P-500 & Jan10-Oct17 & 500 & 94 & 60 & Kaggle.com \\
Industry Portfolio	& IND-30  & Jan90-Jan18 & 30   & 337 & 120 & Kenneth French\\
Industry Portfolio	& IND-49  & Jan90-Jan18 & 49   & 337 & 120 & Kenneth French\\
Cryptocurrency		& C-200  & Apr13-Dec17 & 200  & 57 & 10 & Kaggle.com \\
\bottomrule
\end{tabular}
\caption{\footnotesize{\textbf{Datasets.} Data used for assessment of estimation risk in the empirical studies. The S\&P data is also used for calibrating the simulation study in Section \ref{sec_pf_sim} and for the example in Section \ref{sec_pf_methods}.}}
\label{tab_pf_empirical_datasets}
\end{table}

\subsection{Data and Evaluation Strategy}\label{sec_pf_empirical_data}
\par I assess the ability of ML to reduce estimation risk by considering different real world datasets. The first dataset contains returns on companies from Standard \& Poor's 500 index, which is based on the market capitalisation of 500 of the largest American companies listed on the New York Stock Exchange (NYSE) and NASDAQ. I compute monthly excess returns for $T=334$ observations ranging from January 1990 to October 2017 and consider random draws of $m=20$ and $m=50$ assets. In addition, I consider the full set of $m=500$ assets for a shorter time period. The S\&P data provides an example of estimation risk in large portfolios representative for the US market.

\par Second, I consider two datasets where each asset is constructed by annually assigning each stock from NYSE, American Stock Exchange (AMEX) and NASDAQ to industries based on the Standard Industrial Classification (SIC) codes. The two datasets contain $m=30$ and $m=49$ industries, respectively, and I consider $T=337$ months ranging from January 1990 to January 2018 for each dataset. The data is obtained from the website of Kenneth French. I expect the estimation risk to be less severe in these data, as the aggregation of individual stocks leads to somewhat more stable returns over time.

\par Last, I consider returns on the $m=200$ largest cryptocurrencies by market value as of the end of 2017, observed in $T=57$ months between April 2013 and December 2017. Estimation risk is expected to be large in this dataset due to the relatively short lifetime and the large number of assets. Furthermore, strategies such as the traditional approach and the minimum variance portfolio are impossible to implement due to a degenerate covariance matrix. The data is noisy, and the number of currencies start out at 4 in the first months of 2013, increasing in a close to linear fashion throughout the period. I exclude large monthly returns above 500\% in the absolute sense. All datasets are summarised in Table \ref{tab_pf_empirical_datasets}. 

\par I use a ``rolling-sample'' approach to compare each strategy on a given dataset. In detail, starting from $t=1$, the $n$ first returns are used for estimation, and the first out of sample return at $t=n+1$ is used to compute the portfolio return. One step onwards, estimation is based on $t=2,\hdots,n+1$ and evaluation is conducted using $t=n+2$. Continuing this procedure for all time periods yields $T-n$ out of sample returns. The Sharpe ratio of these returns for strategy $q$ is
\begin{equation}\label{eq_pf_empirical_sharpe}
\hat{s}_q = \frac{\hat{\mu}_q}{\hat{\sigma}_q}
\end{equation}
where the mean $\hat{\mu}_q$ and standard deviation $\hat{\sigma}_q$ are computed based on the $T-n$ out of sample returns for strategy $q$. To test whether the estimated Sharpe ratios of two given strategies are statistically different, I use the approach by \cite{jobson1981performance} with the correction in \cite{memmel2003performance}.\footnote{Suppose the two strategies $q$ and $l$ have out of sample portfolio means $\hat{\mu}_q$, $\hat{\mu}_l$, standard deviations $\hat{\sigma}_q$, $\hat{\sigma}_l$ and covariance $\hat{\sigma}_{ql}$. Under the null hypothesis of equal Sharpe ratios, the test statistic is $\hat{z}_{ql} = \frac{\hat{\sigma}_l \hat{\mu}_q - \hat{\sigma}_q\hat{\mu}_l}{\sqrt{\psi}}$, where $\psi = \frac{1}{T-n}\left(2\hat{\sigma}_q^2\hat{\sigma}_l^2 - 2\hat{\sigma}_q\hat{\sigma}_l \hat{\sigma}_{ql} + \frac{1}{2} \hat{\mu}_q^2\hat{\sigma}_l^2 + \frac{1}{2}\hat{\mu}_l^2\hat{\sigma}_q^2 - \frac{\hat{\mu}_q\hat{\mu}_l}{\hat{\sigma}_q\hat{\sigma}_l}\hat{\sigma}_{ql}^2\right)$.
Based on normally distributed IID returns, the test statistic is asymptotically standard normal.} I use estimation windows $n$ of 120, 60 and 10 months depending on the data, see Table \ref{tab_pf_empirical_datasets}. For Ridge, Lasso and PCR, I use five-fold cross validation for all datasets expect the cryptocurrency data, where I use leave-one-out cross validation. For the Spike and Slab approach, I use the same prior specifications as outlined in the simulation study in Section \ref{sec_pf_sim}.

\subsection{Results}\label{sec_pf_empirical_results}
The Sharpe ratio of each strategy and each dataset is reported in Table \ref{tab_pf_empirical_sharpe}. The results from the S\&P data based on random selections of $m=20$ and $m=50$ assets are reported in the second and third column, respectively. The low Sharpe ratio obtained by the traditional mean-variance approach (``MV'') suggests that the theoretical gains from diversification are eroded by estimation risk. Some details on this finding is provided in Figure \ref{fig_pf_empirical_lasso_mean}, where the out of sample mean of the traditional approach is plotted for each evaluation month. Large changes in asset returns towards the end of the sample causes both the mean and covariance structures to change, giving large asset positions and consequently highly volatile out of sample returns. On the other hand, the cross-validation procedure used by Lasso results in a high penalty in this period, setting all portfolio weights to zero. In other words, the presence of large estimation risk makes the Lasso strategy refrain from investing in risky assets. Similar (although not sparse) results hold for both Ridge and PCR. The result is relatively high Sharpe ratios for the ML strategies in the second and third column of Table \ref{tab_pf_empirical_sharpe}. Furthermore, the equal weight strategy yields negative returns due to a general declining market in the period under study, and no short selling (``MV-C'') or ignoring means (``Min Variance'') does not help in this case. The ML strategies statistically outperform all other strategies in terms of Sharpe ratios, see Table \ref{tab_pf_empirical_jk} in Appendix \ref{sec_pf_appendix_results}.

\begin{table}
\scriptsize
\centering
\begin{tabular}{lrrrrrr}
\midrule
Strategy		& S\&P-20 & S\&P-50 & IND-30 & IND-49 & C-200 & S\&P-500\\
\midrule
MV 				& -0.1 & -0.03 & 0.06 & -0.01 & - & -\\ \\
\textbf{Machine Learning}\\
Ridge			& 0.25 & 0.24 & 0.15 & 0.13 & 0.31 & 0.27\\
Lasso			& 0.24 & 0.23 & 0.08 & 0.10 & 0.32 & 0.20\\
PCR				& 0.21 & 0.24 & 0.16 & 0.08 & 0.22 & 0.19\\ 
Spike \& Slab 	& 0.27 & 0.16 & 0.11 & 0.10 & -0.06 & 0.23\\ \\
\textbf{Benchmarks}\\
Min Variance	& -0.30 & -0.30 & 0.15 & 0.08 & - & -\\
MV-C 			& -0.29 & -0.29 & 0.21 & 0.21 & 0.15 & -0.12\\
Equal weights   & -0.16 & -0.17 & 0.18 & 0.19 & 0.34 & 0.12\\
\bottomrule
\end{tabular}
\caption{\footnotesize{\textbf{Portfolio Sharpe ratio for the empirical data.} The average out of sample Sharpe ratio computed using formula (\ref{eq_pf_empirical_sharpe}) for each strategy and each dataset described in Table \ref{tab_pf_empirical_datasets}. The estimation details are discussed in Section \ref{sec_pf_empirical_data}.}}
\label{tab_pf_empirical_sharpe}
\end{table} 

\par The Sharpe ratios for the industry portfolios are reported in the fourth and fifth columns of Table \ref{tab_pf_empirical_sharpe}. The raw data indicate positive returns for all industries measured across the entire sample. In large, ML outperforms the traditional approach, but the difference is insignificant. Moreover, the equal weight portfolio and the mean-variance portfolio without short selling provide higher Sharpe ratios than the ML approaches, but the difference is not significant. The similar performance across methods could be due to the fact that each asset (industry) is a combination of stocks, to some extent reducing estimation risk compared to portfolios of individual stocks.

\par The results from the cryptocurrency data are reported in column six. In total, 200 assets are considered throughout the analysis, but no more than 60 assets are present during a specific estimation window. Still, estimation is challenging in this case, as these portfolios are chosen based on an estimation period of only 10 months. The main result is that the ML approaches are able to obtain similar Sharpe ratios to the equal weight strategy, but the estimated portfolio weights using Lasso only show a few non-zero assets. Figure \ref{fig_pf_empirical_lasso_nonzero} plots the number of assets in the equal weight strategy together with the number of non-zero assets chosen by Lasso. Broadly speaking, less than 10 assets are used to form the portfolio at each time period in the case of Lasso.

\par The results for the large S\&P portfolio is reported in the last column of Table \ref{tab_pf_empirical_sharpe}. Again, this is a highly challenging estimation problem where the traditional approach and the minmium variance portfolio are infeasible, since the estimation sample size is 60 months and the number of assets is 500. In this high dimensional case, the ML methods outperform, although not signficantly, the equal weight portfolio in terms of Sharpe ratio. This is in contrast to the results found in e.g. \cite{demiguel2007optimal}, where it is argued that the equal weight strategy is superior to the other approaches discussed in the literature. Further results on the average portfolio mean and standard deviations are reported in Table \ref{tab_pf_empirical_mean} and \ref{tab_pf_empirical_std} in Appendix \ref{sec_pf_appendix_results}, respectively.

\begin{figure}
\centering
\begin{minipage}{1\textwidth}
\begin{subfigure}{0.50\textwidth}
\includegraphics[width=\linewidth]{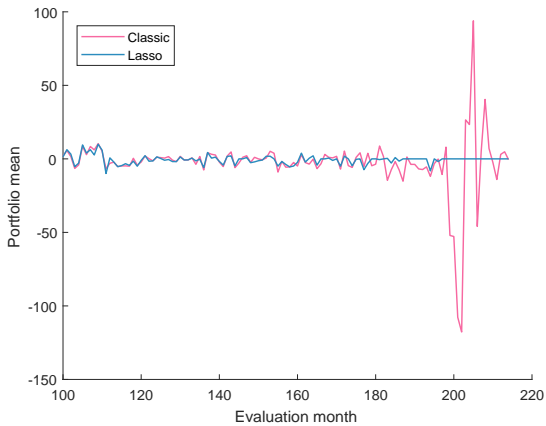}
\caption{Out of sample portfolio mean}
\label{fig_pf_empirical_lasso_mean}
\end{subfigure}\hspace*{\fill}
\begin{subfigure}{0.50\textwidth}
\includegraphics[width=\linewidth]{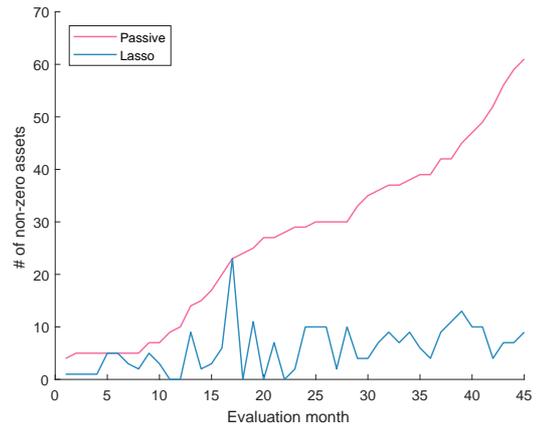}
\caption{Number of non-zero assets}
\label{fig_pf_empirical_lasso_nonzero}
\end{subfigure}\\
\caption{\footnotesize \textbf{Sparse portfolios and estimation risk.} Figure \ref{fig_pf_empirical_lasso_mean} shows the out of sample return from the traditional approach and Lasso based on the S\&P data with $m=50$ assets. The traditional approach yields highly unstable returns towards the end of the sample. In contrast, Lasso chooses a zero investment policy in this period. Figure \ref{fig_pf_empirical_lasso_nonzero} plots the number of non-zero assets from the equal weight passive strategy (all assets) and Lasso based on the cryptocurrency data.}
\label{fig_pf_empirical_lasso}
\end{minipage}
\end{figure}

\section{Conclusion}\label{sec_pf_conclusion}
I provide a unified ML framework for estimation of optimal portfolio weights, where the weights are obtained by ``regressing'' excess asset returns on a constant. The ML framework should be considered by researchers and practitioners for several reasons.

\par First, the framework implies that any ``off-the-shelf'' linear ML method can be used to estimate the optimal portfolio weights. Portfolio optimization via ML possibly simplifies implementation, as such learning algorithms are available in standard statistical software. Consequently, cross-validation and estimation may be done using well documented and standardised software, with several options for diagnostic checks. 

\par Second, the framework can be used to shed new light on the traditional approach and recently proposed shrinkage methods. Since the traditional approach is equivalent to linear regression, the large estimation risk incurred by this method can be interpreted as overfitting in a regression sense. The ML framework can be used to derive a link between the traditional approach and the regularization methods proposed by \cite{demiguel2009generalized}. In particular, I showed how Ridge and Lasso weight estimates relate to the traditional weight estimates, and provided a condition for when Ridge regression may be expected to outperform the traditional approach in terms of estimation risk.

\par Third, I introduce two new methods for portfolio estimation: Principal Component regression and Spike and Slab regression. In essence, I document that these methods perform similarly to Ridge and Lasso regressions based on both simulated and empirical data.

\par Finally, the ML framework provided in this paper offers one promising direction for future research. Specifically, Proposition \ref{prop_pf_framework_generror} can be extended to non-linear models, thus suggesting portfolio estimation based on sophisticated ML algorithms such as Regression trees or Random Forests. However, backing out portfolio weight estimates from such methods is not straightforward. 

\clearpage
\bibliographystyle{aea}
\bibliography{referencelist}

\clearpage
\begin{appendices}

\renewcommand{\thesection}{\Alph{section}}
\numberwithin{equation}{section}

\section{Derivations}\label{sec_pf_appendix_derivations}
\begin{proof}[Proof of Equation (\ref{eq_pf_framework_bv})]
Suppose that $\x_0$ is non-random and $y_0$ is drawn according to $y_0=f(\x_0)+\varepsilon_0$. Using (\ref{eq_pf_framework_mse}) we get
\begin{eqnarray}
F_q &=& \E_{\mathcal{T}}\left\{\E_{y_0}[(y_0-\hat{f}_q(\x_0))^2]\right\}\nonumber\\
&=& \E_{\mathcal{T}}\left\{\E_{y_0}[y_0^2 + \hat{f}_q(\x_0)^2 - 2y_0 \hat{f}_q(\x_0)]\right\}\nonumber\\
&=&\E_{\mathcal{T}}\left\{f(\x_0)^2 + \phi^2 + \hat{f}_q(\x_0)^2 - 2f(\x_0) \hat{f}_q(\x_0)]\right\}\nonumber\\
&=&(f(\x_0)-\E_\mathcal{T}[\hat{f}_q(\x_0)])^2 + \V_\mathcal{T}[\hat{f}_q(\x_0)] + \phi^2
\end{eqnarray}
The last equality follows by using $\E_\mathcal{T}[\hat{f}_q(\x_0)^2] = \V_\mathcal{T}[\hat{f}_q(\x_0)] + \E_{\mathcal{T}}[\hat{f}_q(\x_0)]^2$.
\end{proof}

\begin{proof}[Proof of Equation (\ref{eq_pf_framework_eu})] Let $\tilde{\x}$ be the return on the risky assets. The portfolio return is given by $r_f\theta_f + \tilde{\x}'\boldsymbol{\theta}$, and the optimization problem of the agent is
\begin{equation}
\max_{\theta_f,\boldsymbol{\theta}} \left\{ \E[u(r_f\theta_f + \tilde{\x}'\boldsymbol{\theta})]\right\} \text{ s.t. } \theta_f + \mathbf{1}'\boldsymbol{\theta} = 1
\end{equation}
where $\mathbf{1}$ is a $m\times 1$ vector of ones. Substituting in the constraint gives the optimization problem $\max_{\boldsymbol{\theta}} \left\{ \E[u(r_f + \x'\boldsymbol{\theta})]\right\}$, where excess return is given by $\x = \tilde{\x}-\mathbf{1}r_f$. 
\end{proof}

\begin{proof}[Proof of Proposition \ref{prop_pf_framework_generror}] Apply the following monotonic transformation to (\ref{eq_pf_framework_u})
\begin{equation}
\tilde{u}(r) = 2\frac{1}{\alpha}u(r) - \left(\frac{1}{\alpha}\right)^2 = -\left(\frac{1}{\alpha}-r\right)^2
\end{equation}
Plug in the return $r=r_f+\x'\boldsymbol{\theta}$ to get
\begin{equation}
\tilde{u}(r_f+\x'\boldsymbol{\theta}) = -\left(\frac{1-\alpha r_f}{\alpha}-\x'\boldsymbol{\theta}\right)^2 = -(\bar{r}-\x'\boldsymbol{\theta})^2
\end{equation}
where $\bar{r} = (1-\alpha r_f)/\alpha$. By taking the expectation with respect to $\x$ the generalisation error is defined as
\begin{equation}
F(\boldsymbol{\theta}) = -\E_\x[\tilde{u}(r_f+\x'\boldsymbol{\theta})] = \E_\x[(\bar{r}-\x'\boldsymbol{\theta})^2]
\end{equation} 
Since $\tilde{u}$ is a monotonic transformation of $u$, it directly follows that $\argmax_{\boldsymbol{\theta}}\E_{\x}[u(r_f+ \x'\boldsymbol{\theta})]=\argmin_{\boldsymbol{\theta}} F(\boldsymbol{\theta})$. Thus, the solution $\boldsymbol{\theta}^*$ that maximizes expected quadratic utility is equivalent to the solution that minimizes the generalisation error.
\end{proof}

\begin{proof}[Proof of Equation (\ref{eq_pf_framework_thetaPop})] Writing out the generalisation error (\ref{eq_pf_framework_generror}) using the known distribution of $\x$
\begin{eqnarray}\label{eq_pf_appendix_generror}
F(\boldsymbol{\theta}) &=& (\bar{r}-\boldsymbol{\mu}'\boldsymbol{\theta})^2 + \boldsymbol{\theta}'\boldsymbol{\Sigma}\boldsymbol{\theta}\nonumber\\
&=& \bar{r}^2-2\boldsymbol{\mu}'\boldsymbol{\theta}\bar{r}+\boldsymbol{\theta}'\boldsymbol{\mu}\boldsymbol{\mu}'\boldsymbol{\theta}+\boldsymbol{\theta}'\boldsymbol{\Sigma}\boldsymbol{\theta}\nonumber\\
&=&\bar{r}^2 - 2\boldsymbol{\mu}'\boldsymbol{\theta}\bar{r} + \boldsymbol{\theta}'(\boldsymbol{\Sigma}+\boldsymbol{\mu}\boldsymbol{\mu}')\boldsymbol{\theta}
\end{eqnarray}
Taking the derivative with respect to $\boldsymbol{\theta}$ gives the first order condition
\begin{equation}
-2\boldsymbol{\mu}\bar{r}+2(\boldsymbol{\Sigma}+\boldsymbol{\mu}\boldsymbol{\mu}')\boldsymbol{\theta}=0
\end{equation}
solving for $\boldsymbol{\theta}$ yields
\begin{equation}
\boldsymbol{\theta}^* = (\boldsymbol{\Sigma} + \boldsymbol{\mu}\boldsymbol{\mu}')^{-1} \boldsymbol{\mu} \bar{r}
\end{equation}
which is Equation (\ref{eq_pf_framework_thetaPop}). Alternatively, the same formula can be obtained by maximizing expected quadratic utility $\E_{\x}[u(r_f + \x'\boldsymbol{\theta})] = r_f + \boldsymbol{\mu}'\boldsymbol{\theta} - \frac{1}{2}\alpha \left(\boldsymbol{\theta}'\boldsymbol{\Sigma}\boldsymbol{\theta}+(r_f + \boldsymbol{\mu}'\boldsymbol{\theta})^2\right)$ from (\ref{eq_pf_framework_eu}) directly.
\end{proof}

\begin{proof}[Proof of Proposition \ref{prop_pf_framework_estrisk}] Using (\ref{eq_pf_appendix_generror}), the minimum generalisation error is
\begin{equation}
F_* = F(\boldsymbol{\theta}^*)=\bar{r}^2 - 2\boldsymbol{\theta}^{*'}\boldsymbol{\mu}\bar{r} + \boldsymbol{\theta}^{*'}\mathbf{A}\boldsymbol{\theta}^*
\end{equation}
By using (\ref{eq_pf_framework_thetaPop}) we may substitute in for $\boldsymbol{\mu} \bar{r}=\mathbf{A}\boldsymbol{\theta}^*$, giving
\begin{equation}
F_* = \bar{r}^2 - 2\boldsymbol{\theta}^{*'}\mathbf{A}\boldsymbol{\theta}^* + \boldsymbol{\theta}^{*'}\mathbf{A}\boldsymbol{\theta}^* = \bar{r}^2-\boldsymbol{\theta}^{*'}\mathbf{A}\boldsymbol{\theta}^*
\end{equation}
Similarly, the expected generalisation error of algorithm $q$ can be written using the distribution of the out of sample returns $\x_0\sim\N(\boldsymbol{\mu},\boldsymbol{\Sigma})$ as
\begin{eqnarray}
F_q&=&\E_\mathcal{T}\{\E_{\x_0}[(\bar{r}-\x_0'\hat{\boldsymbol{\theta}}_q)^2]\}\nonumber\\
&=&\E_\mathcal{T}\{\bar{r}^2 - 2\hat{\boldsymbol{\theta}}_q'\boldsymbol{\mu}\bar{r} + \hat{\boldsymbol{\theta}}_q'\mathbf{A}\boldsymbol\hat{\boldsymbol{\theta}}_q\}
\end{eqnarray}
Substituting in for $\boldsymbol{\mu}\bar{r} = \mathbf{A}\boldsymbol{\theta}^*$ now gives
\begin{equation}
F_q = \E_\mathcal{T}\{\bar{r}^2 - 2\hat{\boldsymbol{\theta}}_q' \mathbf{A}\boldsymbol{\theta}^* + \hat{\boldsymbol{\theta}}_q'\mathbf{A}\boldsymbol\hat{\boldsymbol{\theta}}_q\}
\end{equation}
By adding and subtracting $\boldsymbol{\theta}^{*'}\mathbf{A}\boldsymbol{\theta}^*$ we get
\begin{eqnarray}
F_q &=& \bar{r}^2-\boldsymbol{\theta}^{*'}\mathbf{A}\boldsymbol{\theta}^* + \E_\mathcal{T}\{(\boldsymbol{\theta}^*-\hat{\boldsymbol{\theta}}_q)'\mathbf{A}(\boldsymbol{\theta}^*-\hat{\boldsymbol{\theta}}_q)\}\nonumber\\
&=& F_* + \E_\mathcal{T}\{(\boldsymbol{\theta}^*-\hat{\boldsymbol{\theta}}_q)'\mathbf{A}(\boldsymbol{\theta}^*-\hat{\boldsymbol{\theta}}_q)\}
\end{eqnarray}
It thus follows that we may define estimation risk as
\begin{equation}
R_q = F_q-F_* = \E_\mathcal{T}\{(\boldsymbol{\theta}^*-\hat{\boldsymbol{\theta}}_q)'\mathbf{A}(\boldsymbol{\theta}^*-\hat{\boldsymbol{\theta}}_q)\}
\end{equation}
The expectation of the quadratic form is given by the rule $\E[\x'\mathbf{A}\mathbf{x}] =\E[\x]'\mathbf{A}\E[\x]+ \text{tr}(\mathbf{A}\V[\x])$. In our case, define $\x=\boldsymbol{\theta}^*-\hat{\boldsymbol{\theta}}_q$ giving $\E_\mathcal{T}[\x]=\boldsymbol{\theta}^*-\E_\mathcal{T}[\hat{\boldsymbol{\theta}}_q]$ and $\V_\mathcal{T}[\x]=\E_\mathcal{T}[(\boldsymbol{\theta}^*-\hat{\boldsymbol{\theta}}_q-(\boldsymbol{\theta}^*-\E_\mathcal{T}[\hat{\boldsymbol{\theta}}_q]))(\boldsymbol{\theta}^*-\hat{\boldsymbol{\theta}}_q-(\boldsymbol{\theta}^*-\E_\mathcal{T}[\hat{\boldsymbol{\theta}}_q]))']=\E_\mathcal{T}[(\hat{\boldsymbol{\theta}}_q-\E_\mathcal{T}[\hat{\boldsymbol{\theta}}_q])(\hat{\boldsymbol{\theta}}_q-\E_\mathcal{T}[\hat{\boldsymbol{\theta}}_q])']=\V_\mathcal{T}[\hat{\boldsymbol{\theta}}_q]$. Using the expectation rule we thus get
\begin{equation}
R_q = (\boldsymbol{\theta}^*-\E_\mathcal{T}[\hat{\boldsymbol{\theta}}_q])'\mathbf{A}(\boldsymbol{\theta}^*-\E_\mathcal{T}[\hat{\boldsymbol{\theta}}_q]) + \text{tr}\left(\mathbf{A}\V_\mathcal{T}[\hat{\boldsymbol{\theta}}_q]\right) 
\end{equation}
which completes the proof.
\end{proof}

\begin{proof}[Proof of Equation (\ref{eq_pf_framework_tangencyPop})] Similar to \cite{britten1999sampling}, we may expand (\ref{eq_pf_framework_thetaPop}) as
\begin{equation}
\boldsymbol{\theta}^* = (\boldsymbol{\Sigma}+\boldsymbol{\mu}\boldsymbol{\mu}')^{-1} \boldsymbol{\mu}\bar{r} = \left(\boldsymbol{\Sigma}^{-1} - \frac{\boldsymbol{\Sigma}^{-1}\boldsymbol{\mu}\boldsymbol{\mu}'\boldsymbol{\Sigma}^{-1}}{1+\boldsymbol{\mu}'\boldsymbol{\Sigma}^{-1}\boldsymbol{\mu}}\right)\boldsymbol{\mu}\bar{r} = \frac{\bar{r}}{1+\boldsymbol{\mu}'\boldsymbol{\Sigma}^{-1} \boldsymbol{\mu}} \boldsymbol{\Sigma}^{-1}\boldsymbol{\mu}
\end{equation} 
Computing the relative weights gives the tangency portfolio $\boldsymbol{\omega}^* = \frac{\boldsymbol{\theta}^*}{\mathbf{1}'\boldsymbol{\theta}^*} = \frac{\boldsymbol{\Sigma}^{-1}\boldsymbol{\mu}}{\mathbf{1}'\boldsymbol{\Sigma}^{-1}\boldsymbol{\mu}}$.
\end{proof}

\begin{proof}[Proof of Proposition \ref{prop_pf_framework_Ridge}] It is useful to first derive some properties of the regression formulation of the optimal portfolio weights $\bar{r} = \X\boldsymbol{\theta}^* + \varepsilon$, where $\varepsilon$ has mean $\E[\varepsilon]=\bar{r}-\boldsymbol{\mu}'\boldsymbol{\theta}^*$ and variance $\V[\varepsilon] = \boldsymbol{\theta}^{*'}\boldsymbol{\Sigma}\boldsymbol{\theta}^*$. It follows that
\begin{equation}
\phi^2=\E[\varepsilon^2] = \E[\varepsilon]^2+\V[\varepsilon] = (\bar{r}-\boldsymbol{\mu}'\boldsymbol{\theta}^*)^2 + \boldsymbol{\theta}^{*'}\boldsymbol{\Sigma}\boldsymbol{\theta}^*=F_*
\end{equation}
Estimation risk and the corresponding second order moments are
\begin{equation}
R_q = \E_\mathcal{T}\{(\boldsymbol{\theta}^*-\hat{\boldsymbol{\theta}}_q)'\mathbf{A}(\boldsymbol{\theta}^*-\hat{\boldsymbol{\theta}}_q)\} \text{ and } \mathbf{M}_q = \E_\mathcal{T}[(\hat{\boldsymbol{\theta}}_q-\boldsymbol{\theta}^*)(\hat{\boldsymbol{\theta}}_q-\boldsymbol{\theta}^*)']
\end{equation}
Using the results from \cite{theobald1974generalizations}, it follows that if $\mathbf{A}$ is positive semidefinite, $\mathbf{A}\succeq 0$, then $\mathbf{M}_1-\mathbf{M}_2\succeq 0$ if and only if $R_1-R_2\geq 0$ for $q=1,2$. Thus, we may focus on the second order matrix $\mathbf{M}$ in the analysis below. We may write it as
\begin{equation}
\mathbf{M}_q =\V[\hat{\boldsymbol{\theta}}_q] + \E[\hat{\boldsymbol{\theta}}_q-\boldsymbol{\theta}^*]\E[\hat{\boldsymbol{\theta}}_q-\boldsymbol{\theta}^*]'
\end{equation}
When omitting the subscript $q$ I will refer to the traditional approach and when using the subscript ``R'' I refer to Ridge regression. For the traditional (OLS) approach bias is zero, $\E[\hat{\boldsymbol{\theta}}]=\boldsymbol{\theta}^*$, and the variance is $\V[\hat{\boldsymbol{\theta}}] = \phi^2 (\X'\X)^{-1}$, giving $\mathbf{M} =\phi^2(\X'\X)^{-1}$.

\par For Ridge regression, the estimator can be written in terms of the traditional approach as $\hat{\boldsymbol{\theta}}_\text{R} = \mathbf{W}\hat{\boldsymbol{\theta}}$, where $\W=(\I + \lambda(\X'\X)^{-1})^{-1}$. Ridge regression is therefore biased, $\E[\hat{\boldsymbol{\theta}}_\text{R}] = \mathbf{W}\boldsymbol{\theta}^*$, with corresponding variance $\V[\hat{\boldsymbol{\theta}}_\text{R}] = \phi^2\W(\X'\X)^{-1}\W'$, giving
\begin{equation}
\mathbf{M}_\text{R} =\phi^2\W(\X'\X)^{-1}\W' + (\mathbf{W}\boldsymbol{\theta}^*-\boldsymbol{\theta}^*)(\mathbf{W}\boldsymbol{\theta}^*-\boldsymbol{\theta}^*)'
\end{equation}
 From these well known results it follows that
\begin{equation}
\mathbf{M}-\mathbf{M}_\text{R}= \lambda\mathbf{B}^{-1}\left\{2\phi^2\I + \lambda\phi^2(\X'\X)^{-1} - \lambda \boldsymbol{\theta}^*\boldsymbol{\theta}^{*'}\right\}\mathbf{B}^{-1}
\end{equation}
where $\mathbf{B}=\X'\X+\lambda\I$. \cite{theobald1974generalizations} has shown that this expression is positive definite for $\lambda>0$ if $\lambda<2\phi^2/\boldsymbol{\theta}^{*'}\boldsymbol{\theta}^*$ where $\phi^2=F_*$.
\end{proof}

\begin{proof}[Proof of Equation (\ref{eq_pf_methods_existing_empbayes})]
Using Zellner's g-prior $(\mathbf{V}_{\boldsymbol{\eta}}^0)^{-1}=\frac{g}{n}\X_{\boldsymbol{\eta}}'\X_{\boldsymbol{\eta}}$ the posterior variance is
\begin{equation}
\mathbf{V}_{\boldsymbol{\eta}}^1 = (\X_{\boldsymbol{\eta}}'\X_{\boldsymbol{\eta}} + \frac{g}{n}\X_{\boldsymbol{\eta}}'\X_{\boldsymbol{\eta}})^{-1} = \frac{n}{n+g}(\X_{\boldsymbol{\eta}}'\X_{\boldsymbol{\eta}})^{-1}
\end{equation}
The posterior mean follows by using $\X_{\boldsymbol{\eta}}'\y = (\X_{\boldsymbol{\eta}}'\X_{\boldsymbol{\eta}})\hat{\boldsymbol{\theta}}_{\boldsymbol{\eta}}$ as
\begin{equation}
\boldsymbol{\theta}_{\boldsymbol{\eta}}^1 = \frac{n}{n+g}(\X_{\boldsymbol{\eta}}'\X_{\boldsymbol{\eta}})^{-1}((\X_{\boldsymbol{\eta}}'\X_{\boldsymbol{\eta}})\hat{\boldsymbol{\theta}}_{\boldsymbol{\eta}} + \frac{g}{n}(\X_{\boldsymbol{\eta}}'\X_{\boldsymbol{\eta}})\boldsymbol{\theta}_{\boldsymbol{\eta}}^0) = \frac{n}{n+g}\hat{\boldsymbol{\theta}}_{\boldsymbol{\eta}} + \frac{g}{n+g}\boldsymbol{\theta}_{\boldsymbol{\eta}}^0
\end{equation}
\end{proof}

\clearpage
\section{Additional Tables}\label{sec_pf_appendix_results}

\begin{table}[h]
\scriptsize
\centering
\begin{tabular}{lrrrrrrr}
\midrule
Strategy			& MV & Ridge & Lasso & PCR & Spike \& Slab & Min Variance & MV-C\\
\midrule
Ridge			& 0.26***\\
Lasso			& 0.25*** & -0.01\\
PCR				& 0.22** & -0.04 & -0.03\\
Spike \& Slab	& 0.28*** & 0.02 & 0.03 & 0.06\\
Min Variance 	& -0.29** & -0.55*** & -0.54*** & -0.51*** & 0.57***\\
MV-C			& -0.28** & -0.54*** & -0.53*** & -0.50*** & -0.56*** & 0.01\\
Equal weights		& -0.14 & -0.40*** & -0.40*** & -0.36*** & -0.42*** & 0.14* & 0.14**\\
\bottomrule
\end{tabular}
\caption{\footnotesize{\textbf{Jobson and Korkie pairwise test for equal Sharpe ratios.} The table shows that all ML methods (Ridge, Lasso, PCR, Spike \& Slab) significantly outperform the traditional approach (MV) and the benchmark strategies (Min Variance, MV-C and Equal weights) for the case of the S\&P-20 data. The table reports the difference in Sharpe ratios (row strategy $q$ minus column strategy $l$) and the significance level based on the test statistic $\hat{z}_{ql}$ developed by \cite{jobson1981performance}. Stars ***, ** and * indicate significance at the 1\%, 5\% and 10\% levels, respectively.}}
\label{tab_pf_empirical_jk}
\end{table}

\begin{table}[h]
\scriptsize
\centering
\begin{tabular}{lrrrrrr}
\midrule
Strategy		& S\&P-20 & S\&P-50 & IND-30 & IND-49 & C-200 & S\&P-500\\
\midrule
MV 				& -0.20 & -0.46 & 0.63 & -0.07 & - & - \\ \\
\textbf{Machine Learning}\\
Ridge			& 1.03 & 1.04 & 0.76 & 0.62 & 38.18 & 0.73\\
Lasso			& 0.93 & 1.00 & 0.54 & 0.82 & 16.40 & 1.00\\
PCR				& 0.92 & 1.08 & 0.92 & 0.48 & 11.91 & 0.44\\ 
Spike \& Slab 	& 1.01 & 0.68 & 0.45 & 0.40 & -1.55 & 1.05\\ \\
\textbf{Benchmarks}\\
Min Variance	& -1.08 & -1.14 & 0.51 & 0.30 & - & -\\
MV-C 			& -1.02 & -0.98 & 0.69 & 0.68 & 4.05 & -1.36\\
Equal weights   & -0.79 & -0.77 & 0.84 & 0.88 & 12.95 & 0.31\\
\bottomrule
\end{tabular}
\caption{\footnotesize{\textbf{Portfolio means for the empirical data.} The average out of sample portfolio mean, the numerator of formula (\ref{eq_pf_empirical_sharpe}), for each strategy and each dataset described in Table \ref{tab_pf_empirical_datasets}. The estimation details are discussed in Section \ref{sec_pf_empirical_data}.}}
\label{tab_pf_empirical_mean}
\end{table}

\begin{table}[h]
\scriptsize
\centering
\begin{tabular}{lrrrrrr}
\midrule
Strategy		& S\&P-20 & S\&P-50 & IND-30 & IND-49 & C-200 & S\&P-500\\
\midrule
MV 				& 15.25 & 17.05 & 10.69 & 10.12 & - & - \\ \\
\textbf{Machine Learning}\\
Ridge			& 4.15 & 4.41 & 5.03 & 4.71 & 121.64 & 2.78\\
Lasso			& 3.88 & 4.35 & 7.26 & 7.88 & 50.55 & 4.93\\
PCR				& 4.47 & 4.49 & 5.95 & 6.36 & 55.28 & 2.28\\ 
Spike \& Slab 	& 3.75 & 4.33 & 3.95 & 4.00 & 26.90 & 4.65\\ \\
\textbf{Benchmarks}\\
Min Variance	& 3.61 & 3.81 & 3.30 & 3.72 & - & -\\
MV-C 			& 3.50 & 3.42 & 3.28 & 3.33 & 26.50 & 11.39\\
Equal weights  & 5.11 & 4.51 & 4.59 & 4.57 & 37.93 & 2.56\\
\bottomrule
\end{tabular}
\caption{\footnotesize{\textbf{Portfolio standard deviation for the empirical data.} The average out of sample portfolio standard deviation, the denominator of formula (\ref{eq_pf_empirical_sharpe}), for each strategy and each dataset described in Table \ref{tab_pf_empirical_datasets}. The estimation details are discussed in Section \ref{sec_pf_empirical_data}.}}
\label{tab_pf_empirical_std}
\end{table}


\end{appendices}

\end{document}